\documentclass{llncs}

\usepackage{float}
\usepackage{stmaryrd}
\usepackage{graphicx}
\usepackage{xcolor}
\usepackage{amsmath}
\usepackage{amssymb}
\usepackage{multirow}
\usepackage{algcompatible}
\usepackage{algorithmicx}
\usepackage[noend]{algpseudocode}
\usepackage{algorithm}
\usepackage{wrapfig}

\usepackage{physics}
\usepackage{tikz}
\usepackage{mathdots}
\usepackage{yhmath}
\usepackage{cancel}
\usetikzlibrary{patterns}

\usepackage{hyperref}

\usepackage{alltt}

\usepackage{pifont}

\algdef{SE}[DOWHILE]{Do}{doWhile}{\algorithmicdo}[1]{\algorithmicwhile\ #1}%

\makeatletter
\newcommand{\dashedrightarrow}[1][2pt]{%
  \settowidth{\@tempdima}{$\rightarrow$}\rightarrow
  \makebox[-\@tempdima]{\hskip-1.5ex\color{white}\rule[0.5ex]{#1}{1pt}}
  \phantom{\rightarrow}
}
\makeatother

\newcommand{\integers}{\mathbb{Z}}

\newcommand{\peel}{\ensuremath{\mathsf{peel}}}

\newcommand{\ourtool}{\textsc{Diffy}}
\newcommand{\vajra}{\textsc{Vajra}}

\newcommand{\viap}{\textsc{VIAP}}
\newcommand{\veriabs}{\textsc{VeriAbs}}
\newcommand{\zthree}{\textsc{Z3}}

\newcommand{\booster}{\textsc{Booster}}

\newcommand{\vaphor}{\textsc{Vaphor}}

\newcommand{\PP}{\ensuremath{\mathsf{P}}}
\newcommand{\QQ}{\ensuremath{\mathsf{Q}}}
\newcommand{\RR}{\ensuremath{\mathsf{R}}}

\newcommand{\LL}{\ensuremath{\mathsf{L}}}

\newcommand{\EE}{\ensuremath{\mathsf{E}}}

\newcommand{\PB}{\ensuremath{\mathsf{PB}}}
\newcommand{\Stmt}{\ensuremath{\mathsf{St}}}

\newcommand{\scVar}{\ensuremath{v}}
\newcommand{\lpVar}{\ensuremath{\ell}}
\newcommand{\ArVar}{\ensuremath{A}}
\newcommand{\OP}{\ensuremath{\mathsf{op}}}
\newcommand{\BoolE}{\ensuremath{\mathsf{BoolE}}}
\newcommand{\UB}{\ensuremath{\mathsf{UB}}}
\newcommand{\iif}{\ensuremath{\mathbf{if}}}
\newcommand{\eelse}{\ensuremath{\mathbf{else}}}
\newcommand{\tthen}{\ensuremath{\mathbf{then}}}
\newcommand{\ffor}{\ensuremath{\mathbf{for}}}

\newcommand{\cconst}{\ensuremath{\mathsf{c}}}

\newcommand{\true}{\ensuremath{\mathsf{True}}}
\newcommand{\false}{\ensuremath{\mathsf{False}}}
\usepackage{xspace}

\renewcommand{\orcidID}[1]{\href{http://orcid.org/#1}{\raisebox{-1.25pt}{\includegraphics{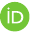}}}}

\usepackage[firstpage]{draftwatermark}

\SetWatermarkText{\hspace*{5in}\raisebox{5.in}{\includegraphics[scale=0.8]{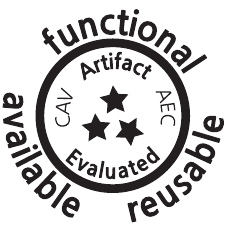}}}

\SetWatermarkAngle{0}
\begin{document}

\title{\textsc{Diffy}: Inductive Reasoning of Array Programs using Difference Invariants}
\author{Supratik Chakraborty\inst{1}\orcidID{0000-0002-7527-7675} \and Ashutosh Gupta\inst{1} \and Divyesh Unadkat\inst{1,2}\orcidID{0000-0001-6106-4719}}
\institute{Indian Institute of Technology Bombay, Mumbai, India\\
  \email{\{supratik,akg\}@cse.iitb.ac.in} \and
  TCS Research, Pune, India\\
  \email{divyesh.unadkat@tcs.com}}

\maketitle

\begin{abstract}
  We present a novel verification technique to prove interesting
properties of a class of array programs with a symbolic parameter $N$
denoting the size of arrays.
The technique relies on constructing two slightly different versions
of the same program.
It infers difference relations between the corresponding variables at
key control points of the joint control-flow graph of the two program
versions.
%
%
The desired post-condition is then proved by
inducting on the program parameter $N$, wherein the difference
invariants are crucially used in the inductive step.
This contrasts with classical techniques that rely on finding
potentially complex loop invaraints for each loop in the program.
Our synergistic combination of inductive reasoning and finding simple
difference invariants helps prove properties of programs that cannot
be proved even by the winner of Arrays sub-category from SV-COMP 2021.
We have implemented a prototype tool called {\ourtool} to demonstrate
these ideas.
We present results comparing the performance of {\ourtool}
with that of state-of-the-art tools.



\end{abstract}

\section{Introduction}
\label{sec:intro}
Software used in a wide range of applications use arrays to store and
update data, often using loops to read and write arrays.  Verifying
correctness properties of such array programs is important, yet
challenging.  A variety of techniques have been proposed in the
literature to address this problem, including inference of quantified
loop invariants~\cite{qli}.  However, it is often difficult to
automatically infer such invariants, especially when programs have
loops that are sequentially composed and/or nested within each other,
and have complex control flows.  This has spurred recent interest in
mathematical induction-based techniques for verifying parametric
properties of array manipulating
programs~\cite{tacas20,viap,sas17,brain}.  While induction-based
techniques are efficient and quite powerful, their Achilles heel is
the automation of the inductive argument.  Indeed, this often becomes
the limiting step in applications of induction-based techniques.
Automating the induction step and expanding the class of
array manipulating programs to which induction-based techniques can be
applied forms the primary motivation for our work.
Rather than being a stand-alone technique, we envisage our work being
used as part of a portfolio of techniques in a modern program verification
tool.

We propose a novel and practically efficient induction-based technique
that advances the state-of-the-art in automating the inductive step
when reasoning about array manipulating programs.  This allows us to
automatically verify interesting properties of a large class of array
manipulating programs that are beyond the reach of state-of-the-art
induction-based techniques, viz.~\cite{tacas20,viap}.  The work that
comes closest to us is {\vajra}~\cite{tacas20}, which is part of the
portfolio of techniques in {\veriabs}~\cite{veriabs} -- the winner of
SV-COMP 2021 in the Arrays Reach sub-category.  Our work addresses
several key limitations of the technique implemented in {\vajra},
thereby making it possible to analyze a much larger class of array
manipulating programs than can be done by {\veriabs}.  Significantly,
this includes programs with nested loops that have hitherto been
beyond the reach of automated techniques that use mathematical
induction~\cite{tacas20,viap,brain}.

A key innovation in our approach is the construction of two slightly
different versions of a given program that have identical control flow
structures but slightly different data operations.  We automatically
identify simple relations, called \emph{difference invariants},
between corresponding variables in the two versions of a program at
key control flow points.  Interestingly, these relations often turn
out to be significantly simpler than inductive invariants required to
prove the property directly. This is not entirely surprising, since
the difference invariants depend less on what individual statements in
the programs are doing, and more on the difference between what they
are doing in the two versions of the program.  We show how the two
versions of a given program can be automatically constructed, and how
differences in individual statements can be analyzed to infer simple
difference invariants.  Finally, we show how these difference
invariants can be used to simplify the reasoning in the inductive step
of our technique.

We consider programs with (possibly nested) loops manipulating arrays,
where the size of each array is a symbolic integer parameter $N~(>
0)$\footnote{For a more general class of programs supported by our
  technique, please see~\cite{diffy-arxiv21}.}.  We verify (a
sub-class of) quantified and quantifier-free properties
that may depend on the symbolic parameter $N$.  Like
in~\cite{tacas20}, we view the verification problem as one of proving
the validity of a parameterized Hoare triple $\{\varphi(N)\} \;\PP_N\;
\{\psi(N)\}$ for all values of $N~(> 0)$, where arrays are of size $N$
in the program $\PP_N$, and $N$ is a free variable in $\varphi(\cdot)$
and $\psi(\cdot)$.

To illustrate the kind of programs that are amenable to our technique,
consider the program shown in Fig.~\ref{fig:motex}(a), adapted from an
SV-COMP benchmark.  This program has a couple of sequentially composed
loops that update arrays and scalars.  The scalars {\tt S} and {\tt F}
are initialized to $0$ and $1$ respectively before the first loop
starts iterating.  Subsequently, the first loop computes a recurrence
in variable {\tt S} and initializes elements of the array {\tt B} to
$1$ if the corresponding elements of array {\tt A} have non-negative
values, and to $0$ otherwise.  The outermost branch condition in the
body of the second loop evaluates to true only if the program
parameter $N$ and the variable {\tt S} have same values.  The value of
{\tt F} is reset based on some conditions depending on corresponding
entries of arrays {\tt A} and {\tt B}.  The pre-condition of this
program is {\tt true}; the post-condition asserts that {\tt F} is
never reset in the second loop.

State-of-the-art techniques find it difficult to prove the assertion
in this program.  Specifically, {\vajra}~\cite{tacas20} is unable to
prove the property, since it cannot reason about the branch condition
(in the second loop) whose value depends on the program parameter $N$.
{\veriabs}~\cite{veriabs}, which employs a sequence of techniques such
as loop shrinking, loop pruning, and inductive reasoning
using~\cite{tacas20} is also unable to verify the assertion shown in
this program.  Indeed, the loops in this program cannot be merged as
the final value of {\tt S} computed by the first loop is required in
the second loop; hence loop shrinking does not help.  Also, loop pruning
does not work due to the complex dependencies in the program and the
fact that the exact value of the recurrence variable {\tt S} is
required to verify the program.  Subsequent abstractions and
techniques applied by {\veriabs} from its portfolio are also unable to
verify the given post-condition.  {\viap}~\cite{viap} translates the
program to a quantified first-order logic formula in the theory of
equality and uninterpreted functions~\cite{viaptheory}. It applies a
sequence of tactics to simplify and prove the generated formula.
These tactics include computing closed forms of recurrences, induction
over array indices and the like to prove the property.  However, its
sequence of tactics is unable to verify this example within our time
limit of $1$ minute.

\begin{figure}[t]
 \begin{tabular}{l|l}
  \begin{minipage}{0.49\textwidth}
  {\scriptsize
\begin{verbatim}
// assume(true)
1. S = 0; F = 1;
2. for(i = 0; i< N; i++) {
3.  S = S + 1;
4.  if ( A[i] >= 0 ) B[i] = 1;
5.  else B[i] = 0;
6. }
7. for(j = 0; j< N; j++) {
8.  if(S == N) {
9.   if ( A[j] >= 0 && !B[j] ) F = 0;
10.  if ( A[j] < 0 && B[j] ) F = 0;
11. }
12.}
// assert(F == 1)
\end{verbatim}
  }
  \begin{center}(a)\end{center}
  \end{minipage}
  &
  \begin{minipage}{0.49\textwidth}
  {\scriptsize
\begin{verbatim}
// assume(true)
1. S = 0;
2. for(i=0; i<N; i++) A[i] = 0;
3. for(j=0; j<N; j++) S = S + 1;
4. for(k=0; k<N; k++) {
5.  for(l=0; l<N; l++) A[l] = A[l] + 1;
6.  A[k] = A[k] + S;
7. }
// assert(forall x in [0,N), A[x]==2*N)
\end{verbatim}
  }
  \begin{center}(b)\end{center}
  \end{minipage}
 \end{tabular}
\caption{Motivating Examples}
\label{fig:motex}
\end{figure}

Benchmarks with nested loops are a long standing challenge for most
verifiers.  Consider the program shown in Fig.~\ref{fig:motex}(b)
with a nested loop in addition to sequentially composed loops.  The first
loop initializes entries in array {\tt A} to $0$.  The second loop
aggregates a constant value in the scalar {\tt S}.  The third loop
is a nested loop that updates array {\tt A} based on the value of
{\tt S}.  The entries of {\tt A} are updated in the inner as
well as outer loop.  The property asserts that on termination, each array
element equals twice the value of the parameter $N$.

While the inductive reasoning of {\vajra} and the tactics in {\viap}
do not support nested loops, the sequence of techniques used by
{\veriabs} is also unable to prove the given post-condition in this
program.  In sharp contrast, our prototype tool {\ourtool} is able to
verify the assertions in both these programs automatically within a
few seconds.  This illustrates the power of the inductive technique
proposed in this paper.

The technical contributions of the paper can be summarized as follows:
\begin{itemize}
  \item We present a novel technique based on mathematical induction
    to prove interesting properties of a class of programs that
    manipulate arrays. The crucial inductive step in our technique
    uses difference invariants from two slightly different versions of
    the same program, and differs significantly from other
    induction-based techniques proposed in the
    literature~\cite{tacas20,viap,sas17,brain}.
  \item We describe algorithms to transform the input program for use
    in our inductive verification technique.  We also present
    techniques to infer simple difference invariants from the two slightly
    different
    program versions, and to complete the inductive step using these
    difference invariants.
  \item We describe a prototype tool {\ourtool} that implements our
    algorithms.
  \item We compare {\ourtool} vis-a-vis state-of-the-art tools for
    verification of C programs that manipulate arrays on a large set
    of benchmarks.  We demonstrate that {\ourtool} significantly
    outperforms the winners of SV-COMP 2019, 2020 and 2021 in the
    Array Reach sub-category.
\end{itemize}



\section{Overview and Relation to Earlier Work}
\label{sec:overview}
In this section, we provide an overview of the main ideas underlying
our technique.  We also highlight how our technique differs
from~\cite{tacas20}, which comes closest to our work.  To keep the
exposition simple, we consider the program $\PP_N$, shown in the first
column of Fig.~\ref{fig:high-lvl}, where $N$ is a symbolic parameter
denoting the sizes of arrays {\tt a} and {\tt b}.  We assume that we
are given a parameterized pre-condition $\varphi(N)$, and our goal is
to establish the parameterized post-condition $\psi(N)$, for all
$N>0$.  In~\cite{tacas20,brain}, techniques based on mathematical
induction (on $N$) were proposed to solve this class of problems.  As
with any induction-based technique, these approaches consist of three
steps.  First, they check if the \emph{base case} holds, i.e. if the
Hoare triple $\{\varphi(N)\}\;\PP_N\;\{\psi(N)\}$ holds for small
values of $N$, say $1 \le N \le M$, for some $M > 0$.  Next, they
assume that the \emph{inductive hypothesis}
$\{\varphi(N-1)\}\;\PP_{N-1}\;\{\psi(N-1)\}$ holds for some $N \ge
M+1$.  Finally, in the \emph{inductive step}, they show that if the
inductive hypothesis holds, so does
$\{\varphi(N)\}\;\PP_N\;\{\psi(N)\}$.  It is not hard to see that the
inductive step is the most crucial step in this style of reasoning.
It is also often the limiting step, since not all programs and
properties allow for efficient inferencing of
$\{\varphi(N)\}\;\PP_N\;\{\psi(N)\}$ from
$\{\varphi(N-1)\}\;\PP_{N-1}\;\{\psi(N-1)\}$.

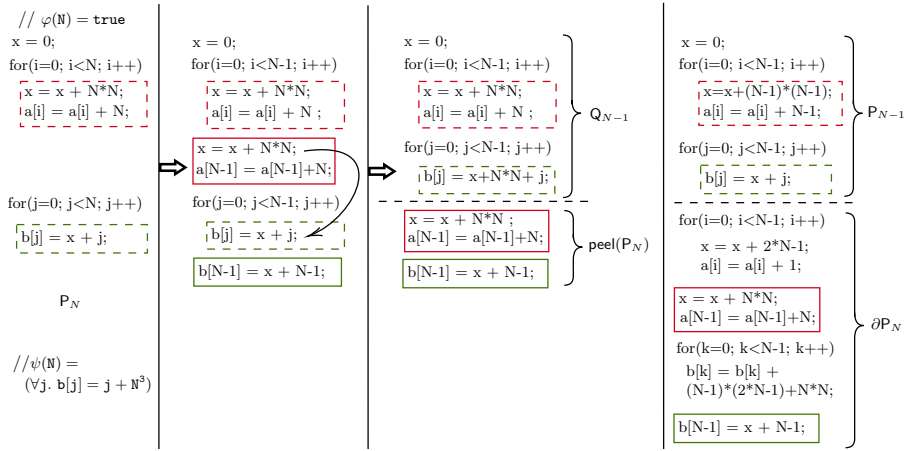
\begin{figure}[!t]

\resizebox{\textwidth}{!}{

\tikzset{every picture/.style={line width=0.75pt}} 

\begin{tikzpicture}[x=0.75pt,y=0.75pt,yscale=-1,xscale=1]

\draw  [color={rgb, 255:red, 208; green, 2; blue, 27 }  ,draw opacity=1 ][dash pattern={on 4.5pt off 4.5pt}] (13.33,57) -- (111.33,57) -- (111.33,92) -- (13.33,92) -- cycle ;
\draw  [color={rgb, 255:red, 65; green, 117; blue, 5 }  ,draw opacity=1 ][dash pattern={on 4.5pt off 4.5pt}] (10.33,168) -- (111.67,168) -- (111.67,190) -- (10.33,190) -- cycle ;
\draw    (120.58,-3) -- (120.58,342) ;
\draw  [line width=1.5]  (121.58,121.15) -- (133.23,121.15) -- (133.23,118.59) -- (141,123.7) -- (133.23,128.8) -- (133.23,126.25) -- (121.58,126.25) -- cycle ;
\draw  [color={rgb, 255:red, 208; green, 2; blue, 27 }  ,draw opacity=1 ][dash pattern={on 4.5pt off 4.5pt}] (158,57) -- (264.33,57) -- (264.33,93.5) -- (158,93.5) -- cycle ;
\draw  [color={rgb, 255:red, 65; green, 117; blue, 5 }  ,draw opacity=1 ][dash pattern={on 4.5pt off 4.5pt}] (157.67,166) -- (264.33,166) -- (264.33,187) -- (157.67,187) -- cycle ;
\draw  [color={rgb, 255:red, 208; green, 2; blue, 27 }  ,draw opacity=1 ] (147,101) -- (257,101) -- (257,136.5) -- (147,136.5) -- cycle ;
\draw  [color={rgb, 255:red, 65; green, 117; blue, 5 }  ,draw opacity=1 ] (147.67,194) -- (262,194) -- (262,215.67) -- (147.67,215.67) -- cycle ;
\draw    (282.58,-2) -- (282.67,342) ;
\draw  [color={rgb, 255:red, 208; green, 2; blue, 27 }  ,draw opacity=1 ][dash pattern={on 4.5pt off 4.5pt}] (322,57) -- (428.33,57) -- (428.33,93) -- (322,93) -- cycle ;
\draw  [color={rgb, 255:red, 65; green, 117; blue, 5 }  ,draw opacity=1 ][dash pattern={on 4.5pt off 4.5pt}] (322.33,123) -- (427.33,123) -- (427.33,144) -- (322.33,144) -- cycle ;
\draw  [color={rgb, 255:red, 208; green, 2; blue, 27 }  ,draw opacity=1 ] (311,157) -- (422.67,157) -- (422.67,189) -- (311,189) -- cycle ;
\draw  [color={rgb, 255:red, 65; green, 117; blue, 5 }  ,draw opacity=1 ] (309.67,195) -- (422,195) -- (422,217) -- (309.67,217) -- cycle ;
\draw  [dash pattern={on 4.5pt off 4.5pt}]  (291,150) -- (448.67,150) -- (454.67,150) ;
\draw   (433.67,147) .. controls (438.34,147) and (440.67,144.67) .. (440.67,140) -- (440.67,91.68) .. controls (440.67,85.01) and (443,81.68) .. (447.67,81.68) .. controls (443,81.68) and (440.67,78.35) .. (440.67,71.68)(440.67,74.68) -- (440.67,28) .. controls (440.67,23.33) and (438.34,21) .. (433.67,21) ;
\draw   (434.67,215) .. controls (439.34,215) and (441.67,212.67) .. (441.67,208) -- (441.67,195.48) .. controls (441.67,188.81) and (444,185.48) .. (448.67,185.48) .. controls (444,185.48) and (441.67,182.15) .. (441.67,175.48)(441.67,178.48) -- (441.67,164.17) .. controls (441.67,159.5) and (439.34,157.17) .. (434.67,157.17) ;
\draw    (233,111) .. controls (301.95,101.15) and (274.85,162.13) .. (237.7,176.38) ;
\draw [shift={(236,177)}, rotate = 341.11] [color={rgb, 255:red, 0; green, 0; blue, 0 }  ][line width=0.75]    (10.93,-3.29) .. controls (6.95,-1.4) and (3.31,-0.3) .. (0,0) .. controls (3.31,0.3) and (6.95,1.4) .. (10.93,3.29)   ;
\draw    (511.58,-4) -- (511.67,342) ;
\draw  [color={rgb, 255:red, 208; green, 2; blue, 27 }  ,draw opacity=1 ][dash pattern={on 4.5pt off 4.5pt}] (539,57) -- (650,57) -- (650,91) -- (539,91) -- cycle ;
\draw  [color={rgb, 255:red, 208; green, 2; blue, 27 }  ,draw opacity=1 ] (520,217) -- (632.67,217) -- (632.67,251) -- (520,251) -- cycle ;
\draw   (651.67,146) .. controls (656.34,146) and (658.67,143.67) .. (658.67,139) -- (658.67,91.73) .. controls (658.67,85.06) and (661,81.73) .. (665.67,81.73) .. controls (661,81.73) and (658.67,78.4) .. (658.67,71.73)(658.67,74.73) -- (658.67,29) .. controls (658.67,24.33) and (656.34,22) .. (651.67,22) ;
\draw   (652.67,341) .. controls (657.34,341) and (659.67,338.67) .. (659.67,334) -- (659.67,256.26) .. controls (659.67,249.59) and (662,246.26) .. (666.67,246.26) .. controls (662,246.26) and (659.67,242.93) .. (659.67,236.26)(659.67,239.26) -- (659.67,166) .. controls (659.67,161.33) and (657.34,159) .. (652.67,159) ;
\draw  [dash pattern={on 4.5pt off 4.5pt}]  (520,151) -- (655.67,151) -- (661.67,151) ;
\draw  [line width=1.5]  (283.58,124.15) -- (295.23,124.15) -- (295.23,121.59) -- (303,126.7) -- (295.23,131.8) -- (295.23,129.25) -- (283.58,129.25) -- cycle ;
\draw  [color={rgb, 255:red, 65; green, 117; blue, 5 }  ,draw opacity=1 ][dash pattern={on 4.5pt off 4.5pt}] (540.33,122) -- (641.67,122) -- (641.67,144) -- (540.33,144) -- cycle ;
\draw  [color={rgb, 255:red, 65; green, 117; blue, 5 }  ,draw opacity=1 ] (519.67,315) -- (632,315) -- (632,337) -- (519.67,337) -- cycle ;

\draw (3,38) node [anchor=north west][inner sep=0.75pt]   [align=left] {for(i=0; i$<$N; i++)};
\draw (2,145) node [anchor=north west][inner sep=0.75pt]   [align=left] {for(j=0; j$<$N; j++)};
\draw (15,58) node [anchor=north west][inner sep=0.75pt]   [align=left] {x = x + N*N;};
\draw (15,73) node [anchor=north west][inner sep=0.75pt]   [align=left] {a[i] = a[i] + N;};
\draw (42,221.4) node [anchor=north west][inner sep=0.75pt]    {$\PP_{N}$};
\draw (12.33,171) node [anchor=north west][inner sep=0.75pt]   [align=left] {b[j] = x + j;};
\draw (144,37) node [anchor=north west][inner sep=0.75pt]   [align=left] {for(i=0; i$<$N-1; i++)};
\draw (145,144) node [anchor=north west][inner sep=0.75pt]   [align=left] {for(j=0; j$<$N-1; j++)};
\draw (160,58) node [anchor=north west][inner sep=0.75pt]   [align=left] {x = x + N*N;};
\draw (160,74) node [anchor=north west][inner sep=0.75pt]   [align=left] {a[i] = a[i] + N ;};
\draw (159.67,169) node [anchor=north west][inner sep=0.75pt]   [align=left] {b[j] = x + j;};
\draw (149,102) node [anchor=north west][inner sep=0.75pt]   [align=left] {x = x + N*N;};
\draw (148,118) node [anchor=north west][inner sep=0.75pt]   [align=left] {a[N-1] = a[N-1]+N;};
\draw (149.67,197) node [anchor=north west][inner sep=0.75pt]   [align=left] {b[N-1] = x + N-1;};
\draw (309,37) node [anchor=north west][inner sep=0.75pt]   [align=left] {for(i=0; i$<$N-1; i++)};
\draw (309,101) node [anchor=north west][inner sep=0.75pt]   [align=left] {for(j=0; j$<$N-1; j++)};
\draw (324,58) node [anchor=north west][inner sep=0.75pt]   [align=left] {x = x + N*N;};
\draw (324,74) node [anchor=north west][inner sep=0.75pt]   [align=left] {a[i] = a[i] + N ;};
\draw (324.33,124) node [anchor=north west][inner sep=0.75pt]   [align=left] {b[j] = x+N*N+ j;};
\draw (313,157) node [anchor=north west][inner sep=0.75pt]   [align=left] {x = x + N*N ;};
\draw (312,171) node [anchor=north west][inner sep=0.75pt]   [align=left] {a[N-1] = a[N-1]+N;};
\draw (311.67,198) node [anchor=north west][inner sep=0.75pt]   [align=left] {b[N-1] = x + N-1;};
\draw (453,76.4) node [anchor=north west][inner sep=0.75pt]    {$\QQ_{N-1}$};
\draw (452,176.4) node [anchor=north west][inner sep=0.75pt]    {$\peel( \PP_{N})$};
\draw (522,37) node [anchor=north west][inner sep=0.75pt]   [align=left] {for(i=0; i$<$N-1; i++)};
\draw (522,102) node [anchor=north west][inner sep=0.75pt]   [align=left] {for(j=0; j$<$N-1; j++)};
\draw (541,58) node [anchor=north west][inner sep=0.75pt]   [align=left] {x=x+(N-1)*(N-1);};
\draw (541,73) node [anchor=north west][inner sep=0.75pt]   [align=left] {a[i] = a[i] + N-1;};
\draw (522,218) node [anchor=north west][inner sep=0.75pt]   [align=left] {x = x + N*N;};
\draw (523,233) node [anchor=north west][inner sep=0.75pt]   [align=left] {a[N-1] = a[N-1]+N;};
\draw (668,74.4) node [anchor=north west][inner sep=0.75pt]    {$\PP_{N-1}$};
\draw (671,239.4) node [anchor=north west][inner sep=0.75pt]    {$\partial \PP_{N}$};
\draw (520,256.92) node [anchor=north west][inner sep=0.75pt]   [align=left] {for(k=0; k$<$N-1; k++)};
\draw (527.67,275) node [anchor=north west][inner sep=0.75pt]   [align=left] {b[k] = b[k] + \\(N-1)*(2*N-1)+N*N;};
\draw (5,20) node [anchor=north west][inner sep=0.75pt]   [align=left] {x = 0;};
\draw (145,21) node [anchor=north west][inner sep=0.75pt]   [align=left] {x = 0;};
\draw (309,20) node [anchor=north west][inner sep=0.75pt]   [align=left] {x = 0;};
\draw (524,21) node [anchor=north west][inner sep=0.75pt]   [align=left] {x = 0;};
\draw (542.33,125) node [anchor=north west][inner sep=0.75pt]   [align=left] {b[j] = x + j;};
\draw (522,158) node [anchor=north west][inner sep=0.75pt]   [align=left] {for(i=0; i$<$N-1; i++)};
\draw (539,178) node [anchor=north west][inner sep=0.75pt]   [align=left] {x = x + 2*N-1;\\a[i] = a[i] + 1;};
\draw (521.67,318) node [anchor=north west][inner sep=0.75pt]   [align=left] {b[N-1] = x + N-1;};
\draw (10,1) node [anchor=north west][inner sep=0.75pt] [font=\normalsize] [align=left] {// {\small $\mathtt{\varphi(N)=true}$}};
\draw (6,268) node [anchor=north west][inner sep=0.75pt] [font=\normalsize] [align=left] {//{\small $\mathtt{\psi(N)=}$}\\
~~{\small $\mathtt{(\forall j.~ b[j]=j+N^{3})}$}};

\end{tikzpicture}

}
\caption{Pictorial Depiction of our Program Transformations}
\label{fig:high-lvl}
\end{figure}
Like in~\cite{tacas20,brain}, our technique uses induction on $N$ to
prove the Hoare triple $\{\varphi(N)\}\;\PP_N\;\{\psi(N)\}$ for all $N
> 0$.  Hence, our base case and inductive hypothesis 
are the same as those in~\cite{tacas20,brain}.  However, our reasoning
in the crucial inductive step is significantly different from that
in~\cite{tacas20,brain}, and this is where our primary contribution
lies.  As we show later, not only does this allow a much larger class
of programs to be efficiently verified compared
to~\cite{tacas20,brain}, it also permits reasoning about classes of
programs with nested loops, that are beyond the reach
of~\cite{tacas20,brain}.  Since the work of~\cite{tacas20}
significantly generalizes that of \cite{brain}, henceforth, we only
refer to~\cite{tacas20} when talking of earlier work that uses
induction on $N$.

In order to better understand our contribution and its difference
vis-a-vis the work of~\cite{tacas20}, a quick recap of the inductive
step used in~\cite{tacas20} is essential.  The inductive step
in~\cite{tacas20} crucially relies on finding a ``difference program''
$\partial \PP_N$ and a ``difference pre-condition'' $\partial
\varphi(N)$ such that: (i) $\PP_N$ is semantically equivalent to
$\PP_{N-1} ; \partial \PP_N$, where ';' denotes sequential composition
of programs\footnote{Although the authors
  of~\cite{tacas20} mention that it suffices to find a $\partial
  \PP_N$ that satisfies $\{\varphi(N)\}\;\PP_{N-1}; \partial \PP_N
  \;\{\psi(N)\}$, they do not discuss any technique that takes
  $\varphi(N)$ or $\psi(N)$ into account when generating $\partial
  \PP_N$.}, (ii) $\varphi(N) \Rightarrow \varphi(N-1)\wedge\partial
\varphi(N)$, and (iii) no variable/array element in $\partial
\varphi(N)$ is modified by $\PP_{N-1}$.  As shown in~\cite{tacas20},
once $\partial \PP_N$ and $\partial \varphi(N)$ satisfying these
conditions are obtained, the problem of proving $\{\varphi(N)\} \;
\PP_N \; \{\psi(N)\}$ can be reduced to that of proving $\{\psi(N-1)
\wedge \partial \varphi(N)\}\;\partial \PP_N\; \{\psi(N)\}$.  This
approach can be very effective if (i) $\partial \PP_N$ is ``simpler''
(e.g. has fewer loops or strictly less deeply nested loops) than
$\PP_N$ and can be computed efficiently, and (ii) a formula $\partial
\varphi(N)$ satisfying the conditions mentioned above exists and can
be computed efficiently.

The requirement of $\PP_N$ being semantically equivalent to
$\PP_{N-1};\partial \PP_N$ is a very stringent one, and finding such a
program $\partial \PP_N$ is non-trivial in general.  In fact, the
authors of~\cite{tacas20} simply provide a set of syntax-guided
conditionally sound heuristics for computing $\partial \PP_N$.
Unfortunately, when these conditions are violated (we have found many
simple programs where they are violated), there are no known
algorithmic techniques to generate $\partial \PP_N$ in a sound manner.
Even if a program $\partial \PP_N$ were to be found in an ad-hoc
manner, it may be as ``complex'' as $\PP_N$ itself.  This makes the
approach of~\cite{tacas20} ineffective for analyzing such programs.
As an example, the fourth column of Fig.~\ref{fig:high-lvl} shows
$\PP_{N-1}$ followed by one possible $\partial \PP_N$ that ensures
$\PP_N$ (shown in the first column of the same figure) is semantically
equivalent to $\PP_{N-1}; \partial \PP_N$.  Notice that $\partial
\PP_N$ in this example has two sequentially composed loops, just like
$\PP_N$ had. In addition, the assignment statement in the body of the
second loop uses a more complex expression than that present in the
corresponding loop of $\PP_N$.  Proving $\{\psi(N-1) \wedge \partial
\varphi(N)\}\;\partial \PP_N\; \{\psi(N)\}$ may therefore not be any
simpler (perhaps even more difficult) than proving $\{\varphi(N)\}\;
\PP_N\;\{\psi(N)\}$.

In addition to the difficulty of computing $\partial \PP_N$, it may be
impossible to find a formula $\partial \varphi(N)$ such that
$\varphi(N) \Rightarrow \varphi(N-1) \wedge \partial \varphi(N)$, as
required by~\cite{tacas20}.  This can happen even for fairly routine
pre-conditions, such as $\varphi(N) \equiv \big(\bigwedge_{i=0}^{N-1}
A[i] = N\big)$.  Notice that there is no $\partial \varphi(N)$ that
satisfies
$\varphi(N) \Rightarrow \varphi(N-1)\wedge \partial \varphi(N)$ in
this case.  In such cases, the technique of~\cite{tacas20} cannot be
used at all, even if $\PP_N$, $\varphi(N)$ and $\psi(N)$ are such that there exists
a trivial proof of $\{\varphi(N)\}\;\PP_N\;\{\psi(N)\}$.

The inductive step proposed in this paper largely mitigates the above
problems, thereby making it possible to efficiently reason about a
much larger class of programs than that possible using the
technique of~\cite{tacas20}.  Our inductive step proceeds as follows.
Given $\PP_N$, we first algorithmically construct two programs
$\QQ_{N-1}$ and $\peel(\PP_N)$, such that $\PP_N$ is semantically
equivalent to $\QQ_{N-1}; \peel(\PP_N)$.  Intuitively, $\QQ_{N-1}$ is
the same as $\PP_N$, but with all loop bounds that depend on $N$ now
modified to depend on $N-1$ instead.  Note that this is different from
$\PP_{N-1}$, which is obtained by replacing \emph{all uses} (not just
in loop bounds) of $N$ in $\PP_N$ by $N-1$.  As we will see, this
simple difference makes the generation of $\peel(\PP_N)$ significantly
simpler than generation of $\partial \PP_N$, as in~\cite{tacas20}.
While generating $\QQ_{N-1}$ and $\peel(P_N)$ may sound similar to
generating $\PP_{N-1}$ and $\partial \PP_N$~\cite{tacas20}, there are
fundamental differences between the two approaches.  First, as noted
above, $\PP_{N-1}$ is semantically different from $\QQ_{N-1}$.
Similarly, $\peel(\PP_N)$ is also semantically different from
$\partial \PP_N$.  Second, we provide an algorithm for generating
$\QQ_{N-1}$ and $\peel(\PP_N)$ that works for a significantly larger
class of programs than that for which the technique of~\cite{tacas20}
works.  Specifically, our algorithm works for all programs amenable to
the technique of~\cite{tacas20}, and also for programs that violate
the restrictions imposed by the grammar and conditional heuristics
in~\cite{tacas20}.  For example, we can algorithmically generate
$\QQ_{N-1}$ and $\peel(\PP_N)$ even for a class of programs with
arbitrarily nested loops -- a program feature explicitly disallowed by
the grammar in~\cite{tacas20}.  Third, we guarantee that
$\peel(\PP_N)$ is ``simpler'' than $\PP_N$ in the sense that the
maximum nesting depth of loops in $\peel(\PP_{N})$ is
\emph{strictly less} than that in $\PP_N$.  Thus, if $\PP_N$ has no
nested loops (all programs amenable to analysis by~\cite{tacas20}
belong to this class), $\peel(\PP_N)$ is guaranteed to be loop-free.
As demonstrated by the fourth column of Fig.~\ref{fig:high-lvl}, no such
guarantees can be given for $\partial \PP_N$ generated by the
technique of~\cite{tacas20}.  This is a significant difference, since
it greatly simplifies the analysis of $\peel(\PP_N)$ vis-a-vis that of
$\partial \PP_N$.

We had mentioned earlier that some pre-conditions $\varphi(N)$ do not
admit any $\partial \varphi(N)$ such that $\varphi(N) \Rightarrow
\varphi(N-1) \wedge \partial \varphi(N)$. It is, however, often easy
to compute formulas $\varphi'(N-1)$ and $\Delta \varphi'(N)$ in such
cases such that $\varphi(N) \Rightarrow \varphi'(N-1) \wedge \Delta
\varphi'(N)$, and the variables/array elements in $\Delta \varphi'(N)$
are not modified by either $\PP_{N-1}$ or $\QQ_{N-1}$.  For example,
if we were to consider a (new) pre-condition $\varphi(N) \equiv \big(\bigwedge_{i=0}^{N-1} A[i] = N\big)$ for
the program $\PP_N$ shown in the first column of
Fig.~\ref{fig:high-lvl}, then we have $\varphi'(N-1) \equiv
\big(\bigwedge_{i=0}^{N-2} A[i] = N\big)$ and $\Delta \varphi'(N) \equiv \big(A[N-1]
= N\big)$.  We assume the availability of such a $\varphi'(N-1)$ and
$\Delta \varphi'(N)$ for the given $\varphi(N)$.  This significantly
relaxes the requirement on pre-conditions and allows a much larger
class of Hoare triples to be proved using our technique vis-a-vis 
that of~\cite{tacas20}.

The third column of Fig.~\ref{fig:high-lvl} shows $\QQ_{N-1}$ and
$\peel(\PP_N)$ generated by our algorithm for the program $\PP_N$ in
the first column of the figure.  It is illustrative to compare these
with $\PP_{N-1}$ and $\partial \PP_N$ shown in the fourth column of
Fig.~\ref{fig:high-lvl}.  Notice that $\QQ_{N-1}$ has the same control
flow structure as $\PP_{N-1}$, but is not semantically equivalent to
$\PP_{N-1}$.  In fact, $\QQ_{N-1}$ and $\PP_{N-1}$ may be viewed as
closely related versions of the same program.  Let $V_\QQ$ and $V_\PP$
denote the set of variables of $\QQ_{N-1}$ and $\PP_{N-1}$
respectively.  We assume $V_\QQ$ is disjoint from $V_\PP$, and analyze
the joint execution of $\QQ_{N-1}$ starting from a state satisfying
the pre-condition $\varphi'(N-1)$, and $\PP_{N-1}$ starting from a
state satisfying $\varphi(N-1)$.  The purpose of this analysis is to
compute a difference predicate $D(V_\QQ, V_\PP, N-1)$ that relates
corresponding variables in $\QQ_{N-1}$ and $\PP_{N-1}$ at the end of
their joint execution.  The above problem is reminiscent of (yet,
different from) translation
validation~\cite{necula00,zuck02,zaks08,barthe11,sharma13,dahiya17,gupta20},
and indeed, our calculation of $D(V_\QQ, V_\PP, N-1)$ is motivated by
techniques from the translation validation literature.  An important
finding of our study is that corresponding variables in $\QQ_{N-1}$
and $\PP_{N-1}$ are often related by simple expressions on $N$,
regardless of the complexity of $\PP_N$, $\varphi(N)$ or $\psi(N)$.
Indeed, in all our experiments, we didn't need to go beyond quadratic
expressions on $N$ to compute $D(V_\QQ, V_\PP, N-1)$.

Once the steps described above are completed, we have
$\Delta \varphi'(N)$, $\peel(\PP_N)$ and $D(V_\QQ, V_\PP, N-1)$.  It
can now be shown that if the inductive hypothesis,
i.e. $\{\varphi(N-1)\} \; \PP_{N-1}\; \{\psi(N-1)\}$ holds, then
proving $\{\varphi(N)\} \;\PP_N\; \{\psi(N)\}$ reduces to proving
$\{\Delta \varphi'(N)
~\wedge~ \psi'(N-1) \} \; \peel(\PP_N) \; \{\psi(N)\}$, where
$\psi'(N-1) \equiv \exists V_\PP \big(\psi(N-1) \wedge D(V_\QQ, V_\PP,
N-1)\big)$.  A few points are worth emphasizing here.  First, if
$D(V_\QQ, V_\PP, N-1)$ is obtained as a set of equalities, the
existential quantifier in the formula $\psi'(N-1)$ can often be
eliminated simply by substitution.  We can also use quantifier
elimination capabilities of modern SMT solvers, viz. Z3~\cite{z3}, to
eliminate the quantifier, if needed.  Second, recall that unlike
$\partial \PP_N$ generated by the technique of~\cite{tacas20},
$\peel(\PP_N)$ is guaranteed to be ``simpler'' than $\PP_N$, and is
indeed loop-free if $\PP_N$ has no nested loops.  Therefore, proving
$\{\Delta \varphi'(N) ~\wedge~ \psi'(N-1)\} \; \peel(\PP_N) \;
\{\psi(N)\}$ is typically significantly simpler than proving
$\{\psi(N-1) \wedge \partial \varphi(N)\} \; \partial \PP_N \;
\{\psi(N)\}$.  Finally, it may happen that the
pre-condition in $\{\Delta \varphi'(N)
~\wedge~ \psi'(N-1)\} \; \peel(\PP_N) \; \{\psi(N)\}$ is not strong
enough to yield a proof of the Hoare triple.  In such cases, we need
to strengthen the existing pre-condition by a formula, say
$\xi'(N-1)$, such that the strengthened pre-condition implies the
weakest pre-condition of $\psi(N)$ under $\peel(\PP_N)$.  Having a
simple structure for $\peel(\PP_N)$ (e.g., loop-free for the entire
class of programs for which~\cite{tacas20} works) makes it
significantly easier to compute the weakest pre-condition.  Note that
$\xi'(N-1)$ is defined over the variables in $V_\QQ$.  In order to
ensure that the inductive proof goes through, we need to strengthen
the post-condition of the original program by $\xi(N)$ such that
$\xi(N-1) \wedge D(V_\QQ, V_\PP, N-1) \Rightarrow \xi'(N-1)$.
Computing $\xi(N-1)$ requires a special form of logical abduction that
ensures that $\xi(N-1)$ refers only to variables in $V_P$.  However,
if $D(V_\QQ, V_\PP, N-1)$ is given as a set of equalities (as is often
the case), $\xi(N-1)$ can be computed from $\xi'(N-1)$ simply by
substitution.  This process of strengthening the pre-condition and
post-condition may need to iterate a few times until a fixed point is
reached, similar to what happens in the inductive step
of~\cite{tacas20}.  Note that the fixed point iterations may not
always converge (verification is undecidable in general).  However, in
our experiments, convergence always happened within a few iterations.
If $\xi'(N-1)$ denotes the formula obtained on reaching the fixed
point, the final Hoare triple to be proved is
$\{\xi'(N-1) \wedge \Delta \varphi'(N)
~\wedge~ \psi'(N-1)\} \; \peel(\PP_N) \; \{\xi(N) \wedge \psi(N)\}$,
where $\psi'(N-1) \equiv \exists V_\PP \big(\psi(N-1) \wedge D(V_\QQ,
V_\PP, N-1)\big)$.  Having a simple (often loop-free) $\peel(\PP_N)$
significantly simplifies the above process.  

We conclude this section by giving an overview of how $\QQ_{N-1}$ and
$\peel(\PP_N)$ are computed for the program $\PP_N$ shown in the first
column of Fig.~\ref{fig:high-lvl}.  The second column of this figure
shows the program obtained from $\PP_N$ by peeling the last iteration
of each loop of the program.  Clearly, the programs in the first and
second columns are semantically equivalent. Since there are no nested
loops in $\PP_N$, the peels (shown in solid boxes) in the second
column are loop-free program fragments.  For each such peel, we
identify variables/array elements modified in the peel and used in
subsequent non-peeled parts of the program.  For example, the variable
{\tt x} is modified in the peel of the first loop and used in the body
of the second loop, as shown by the arrow in the second column of
Fig.~\ref{fig:high-lvl}.  We replace all such uses (if needed,
transitively) by expressions on the right-hand side of assignments in
the peel until no variable/array element modified in the peel is used
in any subsequent non-peeled part of the program. Thus, the use of
{\tt x} in the body of the second loop is replaced by the expression
{\tt x + N*N} in the third column of Fig.~\ref{fig:high-lvl}.
The peeled iteration of the first loop can now be moved to the end of
the program, since the variables modified in this peel are no longer
used in any subsequent non-peeled part of the program.  Repeating the
above steps for the peeled iteration of the second loop, we get the
program shown in the third column of Fig.~\ref{fig:high-lvl}.  This
effectively gives a transformed program that can be divided into two
parts: (i) a program $\QQ_{N-1}$ that differs from $\PP_{N}$ only in
that all loops are truncated to iterate $N-1$ (instead of $N$) times,
and (ii) a program $\peel(\PP_N)$ that is obtained by concatenating
the peels of loops in $\PP_N$ in the same order in which the loops
appeared in $\PP_N$.  It is not hard to see that $\PP_N$, shown in the
first column of Fig.~\ref{fig:high-lvl}, is semantically equivalent to
$\QQ_{N-1}; \peel(\PP_N)$.
Notice that the construction of $\QQ_{N-1}$ and $\peel(\PP_N)$ was
fairly straightforward, and did not require any complex reasoning.  In
sharp contrast, construction of $\partial \PP_N$, as shown in the bottom
half of fourth column of Fig.~\ref{fig:high-lvl}, requires
non-trivial reasoning, and produces a program with two
sequentially composed loops.

\section{Preliminaries and Notation}
\label{sec:prelims}
We consider programs generated by the grammar shown below:

\begin{center}
\begin{tabular}{rcl}
  {\PB} & ::= & \Stmt\\
  \Stmt & ::= & \Stmt~;~\Stmt ~$\mid$~ {\scVar} := \EE ~$\mid$~ {\ArVar}[\EE] := \EE ~$\mid$~
                {\iif}(\BoolE) {\tthen} {\Stmt} {\eelse} \Stmt ~$\mid$~\\
        &     & {\ffor} ({\lpVar} := 0; {\lpVar} $<$ {\UB}; {\lpVar} := {\lpVar}+1) ~\{{\Stmt}\}\\
 \EE    & ::= & \EE ~\OP~ \EE ~$\mid$~ {\ArVar}[\EE] ~$\mid$~ {\scVar} ~$\mid$~ {\lpVar} ~$\mid$~ {\cconst} ~$\mid$~ $N$ \\
 \OP    & ::= & + ~$\mid$~ - ~$\mid$~ * ~$\mid$~ / \\
 \UB    & ::= & \UB ~\OP~ \UB ~$\mid$~ {\lpVar} ~$\mid$~ {\cconst} ~$\mid$~ $N$ \\
 \BoolE & ::= & \EE ~$\mathsf{relop}$ \EE ~$\mid$~ {\BoolE} $\mathsf{AND}$ {\BoolE} ~$\mid$~
                $\mathsf{NOT}$ {\BoolE} ~$\mid$~ {\BoolE} $\mathsf{OR}$ {\BoolE}
\end{tabular}
\end{center}

\noindent
Formally, we consider a program $\PP_N$ to be a tuple $(\mathcal{V},
\mathcal{L}, \mathcal{A}, {\PB}, N)$, where $\mathcal{V}$ is a set of
scalar variables, $\mathcal{L} \subseteq \mathcal{V}$ is a set of
scalar loop counter variables, $\mathcal{A}$ is a set of array
variables, ${\PB}$ is the program body, and $N$ is a special symbol
denoting a positive integer parameter of the program.  In the grammar
shown above, we assume that ${\ArVar} \in \mathcal{A}$, ${\scVar} \in
\mathcal{V} \setminus \mathcal{L}$, ${\lpVar} \in \mathcal{L}$ and
${\cconst}\in \integers$.  We also assume that each loop ${\LL}$ has a
unique loop counter variable $\ell$ that is initialized at the
beginning of ${\LL}$ and is incremented by $1$ at the end of each
iteration.  We assume that the assignments in the body of ${\LL}$ do
not update $\ell$.  For each loop $\LL$ with termination condition
$\ell < \UB$, we require that $\UB$ is an expression in terms of $N$,
variables in $\mathcal{L}$ representing loop counters of loops that
nest $\LL$, and constants as shown in the grammar.
Our grammar allows a large class of programs (with nested loops) to be
analyzed using our technique, and that are beyond the reach of
state-of-the-art tools like~\cite{veriabs,tacas20,viap}.

We verify Hoare triples of the form $\{\varphi(N)\} \;\PP_N\;
\{\psi(N)\}$, where the formulas $\varphi(N)$ and $\psi(N)$ are
either universally quantified formulas of the form $\forall I\,
\left(\alpha(I, N) \Rightarrow \beta(\mathcal{A}, \mathcal{V}, I, N)\right)$
or quantifier-free formulas of the form $\eta(\mathcal{A}, \mathcal{V},
N)$.  In these formulas, $I$ is a sequence of array index variables,
$\alpha$ is a quantifier-free formula in the theory of arithmetic over
integers, and $\beta$ and $\eta$ are quantifier-free formulas in the
combined theory of arrays and arithmetic over integers.
Our technique can also verify a restricted set of existentially
quantified post-conditions.  We give a few illustrative examples
in the Appendix.

For technical reasons, we rename all scalar and array variables in the
program in a pre-processing step as follows.  We rename each scalar
variable using the well-known Static Single Assignment
(SSA)~\cite{ssa} technique, such that the variable is written at (at
most) one location in the program.  We also rename arrays in the
program such that each loop updates its own version of an array and
multiple writes to an array element within the same loop are performed
on different versions of that array.  We use techniques for array
SSA~\cite{arrayssa} renaming studied earlier in the context of
compilers, for this purpose.  In the subsequent exposition, we assume
that scalar and array variables in the program are already SSA renamed,
and that all array and scalar variables referred to in the pre- and
post-conditions are also expressed in terms of SSA renamed
arrays and scalars.



\section{Verification using Difference Invariants}
\label{sec:algorithms}
The key steps in the application of our technique, as discussed in
Section~\ref{sec:overview}, are
\begin{itemize}
\item[A1:] Generation of $\QQ_{N-1}$ and $\peel(\PP_N)$ from a given
  $\PP_N$.
\item[A2:] Generation of $\varphi'(N-1)$ and $\Delta \varphi'(N)$ from
  a given $\varphi(N)$.
\item[A3:] Generation of the difference invariant $D(V_\QQ, V_\PP, N-1)$,
  given $\varphi(N-1)$, $\varphi'(N-1)$, $\QQ_{N-1}$ and $\PP_{N-1}$.
\item[A4:] Proving $\{\Delta \varphi'(N) ~\wedge~ \exists V_\PP
  \big(\psi(N-1) \wedge D(V_\QQ, V_\PP, N-1)\big)\} \; \peel(\PP_N) \;
  \{\psi(N)\}$, possibly by generation of $\xi'(N-1)$ and $\xi(N)$ to
  strengthen the pre- and post-conditions, respectively.
\end{itemize}
We now discuss techniques for solving each of these sub-problems.

\subsection{Generating $\QQ_{N-1}$ and $\peel(\PP_N)$}\label{sec:genprog}

The procedure illustrated in Fig.~\ref{fig:high-lvl} (going from the
first column to the third column) is fairly straightforward if none of
the loops have any nested loops within them.  It is easy to extend
this to arbitrary sequential compositions of non-nested loops.  Having
all variables and arrays in SSA-renamed forms makes it particularly
easy to carry out the substitution exemplified by the arrow shown in
the second column of Fig.~\ref{fig:high-lvl}.  Hence, we don't discuss
any further the generation of $\QQ_{N-1}$ and $\peel(\PP_N)$ when all
loops are non-nested.

\begin{wrapfigure}[9]{r}{0.5\textwidth}
\vspace{-2.5em}

 
\tikzset{
pattern size/.store in=\mcSize, 
pattern size = 5pt,
pattern thickness/.store in=\mcThickness, 
pattern thickness = 0.3pt,
pattern radius/.store in=\mcRadius, 
pattern radius = 1pt}
\makeatletter
\pgfutil@ifundefined{pgf@pattern@name@_1tbtatztq}{
\pgfdeclarepatternformonly[\mcThickness,\mcSize]{_1tbtatztq}
{\pgfqpoint{0pt}{0pt}}
{\pgfpoint{\mcSize+\mcThickness}{\mcSize+\mcThickness}}
{\pgfpoint{\mcSize}{\mcSize}}
{
\pgfsetcolor{\tikz@pattern@color}
\pgfsetlinewidth{\mcThickness}
\pgfpathmoveto{\pgfqpoint{0pt}{0pt}}
\pgfpathlineto{\pgfpoint{\mcSize+\mcThickness}{\mcSize+\mcThickness}}
\pgfusepath{stroke}
}}
\makeatother

 
\tikzset{
pattern size/.store in=\mcSize, 
pattern size = 5pt,
pattern thickness/.store in=\mcThickness, 
pattern thickness = 0.3pt,
pattern radius/.store in=\mcRadius, 
pattern radius = 1pt}
\makeatletter
\pgfutil@ifundefined{pgf@pattern@name@_zyc6vky2c}{
\pgfdeclarepatternformonly[\mcThickness,\mcSize]{_zyc6vky2c}
{\pgfqpoint{-\mcThickness}{-\mcThickness}}
{\pgfpoint{\mcSize}{\mcSize}}
{\pgfpoint{\mcSize}{\mcSize}}
{
\pgfsetcolor{\tikz@pattern@color}
\pgfsetlinewidth{\mcThickness}
\pgfpathmoveto{\pgfpointorigin}
\pgfpathlineto{\pgfpoint{0}{\mcSize}}
\pgfusepath{stroke}
}}
\makeatother

 
\tikzset{
pattern size/.store in=\mcSize, 
pattern size = 5pt,
pattern thickness/.store in=\mcThickness, 
pattern thickness = 0.3pt,
pattern radius/.store in=\mcRadius, 
pattern radius = 1pt}
\makeatletter
\pgfutil@ifundefined{pgf@pattern@name@_wh328fu5g}{
\pgfdeclarepatternformonly[\mcThickness,\mcSize]{_wh328fu5g}
{\pgfqpoint{0pt}{-\mcThickness}}
{\pgfpoint{\mcSize}{\mcSize}}
{\pgfpoint{\mcSize}{\mcSize}}
{
\pgfsetcolor{\tikz@pattern@color}
\pgfsetlinewidth{\mcThickness}
\pgfpathmoveto{\pgfqpoint{0pt}{\mcSize}}
\pgfpathlineto{\pgfpoint{\mcSize+\mcThickness}{-\mcThickness}}
\pgfusepath{stroke}
}}
\makeatother
\tikzset{every picture/.style={line width=0.75pt}} 

\resizebox{0.5\textwidth}{!}{

\begin{tikzpicture}[x=0.75pt,y=0.75pt,yscale=-1,xscale=1]

\draw  [pattern=_1tbtatztq,pattern size=6pt,pattern thickness=0.75pt,pattern radius=0pt, pattern color={rgb, 255:red, 0; green, 0; blue, 0}] (55.5,24) -- (211.5,24) -- (211.5,44.83) -- (55.5,44.83) -- cycle ;
\draw  [pattern=_zyc6vky2c,pattern size=6pt,pattern thickness=0.75pt,pattern radius=0pt, pattern color={rgb, 255:red, 0; green, 0; blue, 0}] (78,72.37) .. controls (78,69.96) and (79.96,68) .. (82.37,68) -- (222.8,68) .. controls (225.21,68) and (227.17,69.96) .. (227.17,72.37) -- (227.17,85.47) .. controls (227.17,87.88) and (225.21,89.83) .. (222.8,89.83) -- (82.37,89.83) .. controls (79.96,89.83) and (78,87.88) .. (78,85.47) -- cycle ;
\draw  [pattern=_wh328fu5g,pattern size=6pt,pattern thickness=0.75pt,pattern radius=0pt, pattern color={rgb, 255:red, 0; green, 0; blue, 0}] (55.5,98) -- (211.5,98) -- (211.5,118.83) -- (55.5,118.83) -- cycle ;
\draw   (54.42,49.25) .. controls (49.75,49.25) and (47.42,51.58) .. (47.42,56.25) -- (47.42,58.75) .. controls (47.42,65.42) and (45.09,68.75) .. (40.42,68.75) .. controls (45.09,68.75) and (47.42,72.08) .. (47.42,78.75)(47.42,75.75) -- (47.42,81.25) .. controls (47.42,85.92) and (49.75,88.25) .. (54.42,88.25) ;
\draw   (29,5.25) .. controls (24.33,5.25) and (22,7.58) .. (22,12.25) -- (22,54.75) .. controls (22,61.42) and (19.67,64.75) .. (15,64.75) .. controls (19.67,64.75) and (22,68.08) .. (22,74.75)(22,71.75) -- (22,117.25) .. controls (22,121.92) and (24.33,124.25) .. (29,124.25) ;

\draw (41,4) node [anchor=north west][inner sep=0.75pt]   [align=left] {for($\displaystyle \ell _{1}$=0; $\displaystyle \ell _{1}$$\displaystyle < $$\displaystyle N$; $\displaystyle \ell _{1}$++)};
\draw (57.5,48.83) node [anchor=north west][inner sep=0.75pt]   [align=left] {for($\displaystyle \ell _{2}$=0; $\displaystyle \ell _{2}$$\displaystyle < $$\displaystyle N$; $\displaystyle \ell _{2}$++)};
\draw (26,63) node [anchor=north west][inner sep=0.75pt]   [align=left] {$\displaystyle \mathsf{L}_{2}$};
\draw (0,61) node [anchor=north west][inner sep=0.75pt]   [align=left] {$\displaystyle \mathsf{L}_{1}$};


\draw (215,26.25) node [anchor=north west][inner sep=0.75pt]   [align=left] {B1};
\draw (230,71.25) node [anchor=north west][inner sep=0.75pt]   [align=left] {B2};
\draw (214,101.25) node [anchor=north west][inner sep=0.75pt]   [align=left] {B3};

\end{tikzpicture}
}
\caption{A Generic Nested Loop}
\label{fig:nested}
\end{wrapfigure}
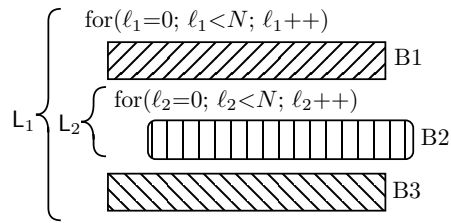

The case of nested loops is, however, challenging and requires
additional discussion.  Before we present an algorithm for handling
this case, we discuss the intuition using an abstract example.
Consider a pair of nested loops, $\LL_1$ and $\LL_2$, as shown in
Fig.~\ref{fig:nested}.  Suppose that {\tt B1} and {\tt B3} are
loop-free code fragments in the body of $\LL_1$ that precede and
succeed the nested loop $\LL_2$.  Suppose further that the loop body,
{\tt B2}, of $\LL_2$ is loop-free.  To focus on the key aspects of
computing peels of nested loops, we make two simplifying assumptions:
(i) no scalar variable or array element modified in {\tt B2} is used
subsequently (including transitively) in either {\tt B3} or {\tt B1},
and (ii) every scalar variable or array element that is modified in
{\tt B1} and used subsequently in {\tt B2}, is not modified again in
either {\tt B1}, {\tt B2} or {\tt B3}.  Note that these assumptions
are made primarily to simplify the exposition.  For a detailed
discussion on how our technique can be used even with some relaxations
of these assumptions, the reader is referred to~\cite{diffy-arxiv21}.
The peel of the abstract loops $\LL_1$ and $\LL_2$ is as shown in
Fig.~\ref{fig:peel-nested}.  The first loop in the peel includes the
last iteration of $\LL_2$ in each of the $N-1$ iterations of $\LL_1$,
that was missed in $\QQ_{N-1}$. The subsequent code includes the last
iteration of $\LL_1$ that was missed in $\QQ_{N-1}$.

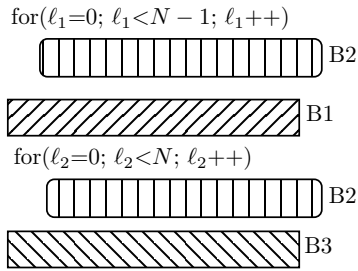
\begin{wrapfigure}[11]{r}{0.4\textwidth}
\vspace{-2em}

 
\tikzset{
pattern size/.store in=\mcSize, 
pattern size = 5pt,
pattern thickness/.store in=\mcThickness, 
pattern thickness = 0.3pt,
pattern radius/.store in=\mcRadius, 
pattern radius = 1pt}
\makeatletter
\pgfutil@ifundefined{pgf@pattern@name@_3vxwom4jt}{
\pgfdeclarepatternformonly[\mcThickness,\mcSize]{_3vxwom4jt}
{\pgfqpoint{-\mcThickness}{-\mcThickness}}
{\pgfpoint{\mcSize}{\mcSize}}
{\pgfpoint{\mcSize}{\mcSize}}
{
\pgfsetcolor{\tikz@pattern@color}
\pgfsetlinewidth{\mcThickness}
\pgfpathmoveto{\pgfpointorigin}
\pgfpathlineto{\pgfpoint{0}{\mcSize}}
\pgfusepath{stroke}
}}
\makeatother

 
\tikzset{
pattern size/.store in=\mcSize, 
pattern size = 5pt,
pattern thickness/.store in=\mcThickness, 
pattern thickness = 0.3pt,
pattern radius/.store in=\mcRadius, 
pattern radius = 1pt}
\makeatletter
\pgfutil@ifundefined{pgf@pattern@name@_gbcn3umkl}{
\pgfdeclarepatternformonly[\mcThickness,\mcSize]{_gbcn3umkl}
{\pgfqpoint{0pt}{0pt}}
{\pgfpoint{\mcSize+\mcThickness}{\mcSize+\mcThickness}}
{\pgfpoint{\mcSize}{\mcSize}}
{
\pgfsetcolor{\tikz@pattern@color}
\pgfsetlinewidth{\mcThickness}
\pgfpathmoveto{\pgfqpoint{0pt}{0pt}}
\pgfpathlineto{\pgfpoint{\mcSize+\mcThickness}{\mcSize+\mcThickness}}
\pgfusepath{stroke}
}}
\makeatother

 
\tikzset{
pattern size/.store in=\mcSize, 
pattern size = 5pt,
pattern thickness/.store in=\mcThickness, 
pattern thickness = 0.3pt,
pattern radius/.store in=\mcRadius, 
pattern radius = 1pt}
\makeatletter
\pgfutil@ifundefined{pgf@pattern@name@_q5ukuj8r9}{
\pgfdeclarepatternformonly[\mcThickness,\mcSize]{_q5ukuj8r9}
{\pgfqpoint{-\mcThickness}{-\mcThickness}}
{\pgfpoint{\mcSize}{\mcSize}}
{\pgfpoint{\mcSize}{\mcSize}}
{
\pgfsetcolor{\tikz@pattern@color}
\pgfsetlinewidth{\mcThickness}
\pgfpathmoveto{\pgfpointorigin}
\pgfpathlineto{\pgfpoint{0}{\mcSize}}
\pgfusepath{stroke}
}}
\makeatother

 
\tikzset{
pattern size/.store in=\mcSize, 
pattern size = 5pt,
pattern thickness/.store in=\mcThickness, 
pattern thickness = 0.3pt,
pattern radius/.store in=\mcRadius, 
pattern radius = 1pt}
\makeatletter
\pgfutil@ifundefined{pgf@pattern@name@_uo9d5uh0d}{
\pgfdeclarepatternformonly[\mcThickness,\mcSize]{_uo9d5uh0d}
{\pgfqpoint{0pt}{-\mcThickness}}
{\pgfpoint{\mcSize}{\mcSize}}
{\pgfpoint{\mcSize}{\mcSize}}
{
\pgfsetcolor{\tikz@pattern@color}
\pgfsetlinewidth{\mcThickness}
\pgfpathmoveto{\pgfqpoint{0pt}{\mcSize}}
\pgfpathlineto{\pgfpoint{\mcSize+\mcThickness}{-\mcThickness}}
\pgfusepath{stroke}
}}
\makeatother
\tikzset{every picture/.style={line width=0.75pt}} 

\resizebox{0.4\textwidth}{!}{

\begin{tikzpicture}[x=0.75pt,y=0.75pt,yscale=-1,xscale=1]

\draw  [pattern=_3vxwom4jt,pattern size=6pt,pattern thickness=0.75pt,pattern radius=0pt, pattern color={rgb, 255:red, 0; green, 0; blue, 0}] (25,31.37) .. controls (25,28.96) and (26.96,27) .. (29.37,27) -- (183.05,27) .. controls (185.46,27) and (187.42,28.96) .. (187.42,31.37) -- (187.42,44.47) .. controls (187.42,46.88) and (185.46,48.83) .. (183.05,48.83) -- (29.37,48.83) .. controls (26.96,48.83) and (25,46.88) .. (25,44.47) -- cycle ;
\draw  [pattern=_gbcn3umkl,pattern size=6pt,pattern thickness=0.75pt,pattern radius=0pt, pattern color={rgb, 255:red, 0; green, 0; blue, 0}] (6.5,62) -- (174.42,62) -- (174.42,82.83) -- (6.5,82.83) -- cycle ;
\draw  [pattern=_q5ukuj8r9,pattern size=6pt,pattern thickness=0.75pt,pattern radius=0pt, pattern color={rgb, 255:red, 0; green, 0; blue, 0}] (29,112.37) .. controls (29,109.96) and (30.96,108) .. (33.37,108) -- (183.05,108) .. controls (185.46,108) and (187.42,109.96) .. (187.42,112.37) -- (187.42,125.47) .. controls (187.42,127.88) and (185.46,129.83) .. (183.05,129.83) -- (33.37,129.83) .. controls (30.96,129.83) and (29,127.88) .. (29,125.47) -- cycle ;
\draw  [pattern=_uo9d5uh0d,pattern size=6pt,pattern thickness=0.75pt,pattern radius=0pt, pattern color={rgb, 255:red, 0; green, 0; blue, 0}] (6.5,138) -- (174.42,138) -- (174.42,158.83) -- (6.5,158.83) -- cycle ;

\draw (7,7) node [anchor=north west][inner sep=0.75pt]   [align=left] {for($\displaystyle \ell _{1}$=0; $\displaystyle \ell _{1}$$\displaystyle < $$\displaystyle N-1$; $\displaystyle \ell _{1}$++)};
\draw (8.5,87.83) node [anchor=north west][inner sep=0.75pt]   [align=left] {for($\displaystyle \ell _{2}$=0; $\displaystyle \ell _{2}$$\displaystyle < $$\displaystyle N$; $\displaystyle \ell _{2}$++)};


\draw (190,30) node [anchor=north west][inner sep=0.75pt]   [align=left] {B2};
\draw (190,111) node [anchor=north west][inner sep=0.75pt]   [align=left] {B2};
\draw (177,64) node [anchor=north west][inner sep=0.75pt]   [align=left] {B1};
\draw (176,140) node [anchor=north west][inner sep=0.75pt]   [align=left] {B3};

\end{tikzpicture}
}
\caption{Peel of the Nested Loop}
\label{fig:peel-nested}
\end{wrapfigure}

Formally, we use the notation $\LL_1${\tt(N)} to denote a loop $\LL_1$
that has no nested loops within it, and its loop counter, say
$\ell_1$, increases from $0$ to an upper bound that is given by an
expression in $N$.
Similarly, we use {\tt $\LL_1$(N, $\LL_2$(N))} to denote a loop
$\LL_1$ that has another loop $\LL_2$ nested within it.  The loop
counter $\ell_1$ of $\LL_1$ increases from $0$ to an upper bound
expression in $N$, while the loop counter $\ell_2$ of $\LL_2$
increases from $0$ to an upper bound expression in $\ell_1$ and $N$.
Using this notation, {\tt $\LL_1$(N, $\LL_2$(N, $\LL_3$(N)))}
represents three nested loops, and so on.  Notice that the upper bound
expression for a nested loop can depend not only on $N$ but also on
the loop counters of other loops nesting it. For notational clarity,
we also use {\tt LPeel($\LL_i$, a, b)} to denote the peel of loop
$\LL_i$ consisting of all iterations of $\LL_i$ where the value of
$\ell_i$ ranges from {\tt a} to {\tt b-1}, both inclusive. Note that
if {\tt b-a} is a constant, this corresponds to the concatenation of
{\tt (b-a)} peels of $\LL_i$.

\begin{wrapfigure}[5]{r}{0.47\textwidth}
\vspace{-3em}
{\footnotesize
\begin{alltt}
for(\(\ell\sb{1}\)=0; \(\ell\sb{1}\)<\(U\sb{\LL\sb{1}}\)(N-1); \(\ell\sb{1}\)++)
  LPeel(\(\LL\sb{2}\), \(U\sb{\LL\sb{2}}\)(\(\ell\sb{1}\),N-1), \(U\sb{\LL\sb{2}}\)(\(\ell\sb{1}\),N))
LPeel(\(\LL\sb{1}\), \(U\sb{\LL\sb{1}}\)(N-1), \(U\sb{\LL\sb{1}}\)(N))
\end{alltt}
}
\vspace{-1em}
\caption{Peel of {\tt $\LL_1$(N, $\LL_2$(N))}}
\label{fig:peell1l2}
\end{wrapfigure}

We will now try to see how we can implement the transformation from
the first column to the second column of Fig.~\ref{fig:high-lvl} for a
nested loop {\tt $\LL_1$(N, $\LL_2$(N))}.  The first step is to
truncate all loops to use $N-1$ instead of $N$ in the upper bound
expressions.  Using the notation introduced above, this gives the loop
{\tt $\LL_1$(N-1, $\LL_2$(N-1))}.  Note that all uses of $N$ other
than in loop upper bound expressions stay unchanged as we go from {\tt
  $\LL_1$(N, $\LL_2$(N))} to {\tt $\LL_1$(N-1, $\LL_2$(N-1))}. We now
ask: \emph{Which are the loop iterations of {\tt $\LL_1$(N,
    $\LL_2$(N))} that have been missed (or skipped) in going to {\tt
    $\LL_1$(N-1, $\LL_2$(N-1))}?}  Let the upper bound expression of
$\LL_1$ in {\tt $\LL_1$(N, $\LL_2$(N))} be $U_{\LL_1}(N)$, and that of
$\LL_2$ be $U_{\LL_2}(\ell_1, N)$.  It is not hard to see that in
every iteration $\ell_1$ of $\LL_1$, where $0 \le \ell_1 <
U_{\LL_1}(N-1)$, the iterations corresponding to $\ell_2 \in
\{U_{\LL_2}(\ell_1, N-1), \ldots, U_{\LL_2}(\ell_1, N)-1\}$ have been
missed.  In addition, all iterations of $\LL_1$ corresponding to
$\ell_1 \in \{U_{\LL_1}(N-1), \ldots, U_{\LL_1}(N)-1\}$ have also been
missed.  This implies that the ``peel'' of {\tt $\LL_1$(N,
  $\LL_2$(N))} must include all the above missed iterations. This peel
therefore is the program fragment shown in Fig.~\ref{fig:peell1l2}.

\begin{wrapfigure}[8]{r}{0.57\textwidth}
\vspace{-3em}
{\footnotesize
\begin{alltt}
for(\(\ell\sb{1}\)=0; \(\ell\sb{1}\)<\(U\sb{\LL\sb{1}}\)(N-1); \(\ell\sb{1}\)++) \{
  for(\(\ell\sb{2}\)=0; \(\ell\sb{2}\)<\(U\sb{\LL\sb{2}}\)(\(\ell\sb{1}\),N-1); \(\ell\sb{2}\)++)
    LPeel(\(\LL\sb{3}\), \(U\sb{\LL\sb{3}}\)(\(\ell\sb{1}\),\(\ell\sb{2}\),N-1), \(U\sb{\LL\sb{3}}\)(\(\ell\sb{1}\),\(\ell\sb{2}\),N))
  LPeel(\(\LL\sb{2}\), \(U\sb{\LL\sb{2}}\)(\(\ell\sb{1}\),N-1), \(U\sb{\LL\sb{2}}\)(\(\ell\sb{1}\),N))
\}
LPeel(\(\LL\sb{1}\), \(U\sb{\LL\sb{1}}\)(N-1), \(U\sb{\LL\sb{1}}\)(N))
\end{alltt}
}
\vspace{-1em}
\caption{Peel of {\tt $\LL_1$(N, $\LL_2$(N, $\LL_3$(N)))}}
\label{fig:peell1l2l3}
\end{wrapfigure}

Notice that if {\tt $U_{\LL_2}$($\ell_1$,N) -
  $U_{\LL_2}$($\ell_1$,N-1)} is a constant (as is the case if {\tt
  $U_{\LL_2}$($\ell_1$,N)} is any linear function of $\ell_1$ and
$N$), then the peel does not have any loop with nesting depth 2.
Hence, the maximum nesting depth of loops in the peel is strictly less
than that in {\tt $\LL_1$(N, $\LL_2$(N))}, yielding a peel that is
``simpler'' than the original program.
This argument can be easily generalized to loops with arbitrarily
large nesting depths.  The peel of {\tt $\LL_1$(N, $\LL_2$(N,
  $\LL_3$(N)))} is as shown in Fig.~\ref{fig:peell1l2l3}.

\begin{wrapfigure}[7]{r}{0.56\textwidth}
  \vspace{-2em}
  \begin{tabular}{ccc}
  \begin{minipage}{0.27\textwidth}
    {\footnotesize
\begin{alltt}
for(i=0; i<N; i++)
  for(j=0; j<N; j++)
    A[i][j] = N;
  
\end{alltt}
    }
  \vspace{-2em}
  \begin{center}(a)\end{center}
  \end{minipage}
  &&
  \begin{minipage}{0.28\textwidth}
    {\footnotesize
\begin{alltt}
for(i=0; i<N-1; i++)
  A[i][N-1] = N;
for(j=0; j<N; j++)
  A[N-1][j] = N;
\end{alltt}
    }
  \vspace{-2em}
  \begin{center}(b)\end{center}
  \end{minipage}
  \end{tabular}
  \caption{(a) Nested Loop \& (b) Peel}
  \label{fig:simpl-nested}
\end{wrapfigure}

As an illustrative example, let us consider the program in
Fig.~\ref{fig:simpl-nested}(a), and suppose we wish to compute the
peel of this program containing nested loops.  In this case, the upper
bounds of the loops are $U_{\LL_1}(N) = U_{\LL_2}(N) = N$.  The peel
is shown in Fig.~\ref{fig:simpl-nested}(b) and consists of two
sequentially composed non-nested loops.  The first loop takes into
account the missed iterations of the inner loop (a single iteration
in this example) that are executed in $\PP_N$ but are missed in
$\QQ_{N-1}$.  The second loop takes into account the missed iterations
of the outer loop in $\QQ_{N-1}$ compared to $\PP_N$.

\begin{algorithm}[!t]
  \caption{\footnotesize \textsc{GenQandPeel}({$\PP_N$}: program)}
  \label{alg:transform}
  \scriptsize
  \begin{algorithmic}[1]
    \State Let sequentially composed loops in $\PP_N$ be in the order $\LL_1$, $\LL_2$, $\ldots$, $\LL_m$;
    \For{each loop $\LL_i \in \textsc{TopLevelLoops}( \PP_N )$}\label{line:top-lvl}
      \State $\langle\QQ_{\LL_i}, \RR_{\LL_i}\rangle \leftarrow$ \textsc{GenQandPeelForLoop}($\LL_i$); \label{line:genlqr}
    \EndFor
    \While{$\exists v. {\it use(v)} \in \QQ_{\LL_i}$ $\wedge$ {\it def(v)} $\in \RR_{\LL_j}$, for some $1 \leq j < i \leq N$}\label{line:repairql1}\label{line:repairq} \Comment{$v$ is var/array element}
      \State Substitute rhs expression for $v$ from $\RR_{\LL_j}$ in $\QQ_{\LL_i}$; \label{line:repairql2} \Comment{If $\RR_{\LL_j}$ is a loop, abort}
    \EndWhile
    \State $\QQ_{N-1} \leftarrow \QQ_{\LL_1};\QQ_{\LL_2};\ldots;\QQ_{\LL_m}$; 
    \State $\peel(\PP_{N}) \leftarrow \RR_{\LL_1};\RR_{\LL_2};\ldots;\RR_{\LL_m}$;
    \State \Return $\langle\QQ_{N-1}, \peel(\PP_N)\rangle$;\label{line:returnqpeel}
    \vspace{2mm}
    \Procedure{\textsc{GenQandPeelForLoop}}{$\LL$: loop}
      \State Let $U_{\LL}(N)$ be the $\UB$ expression  of loop $\LL$;
      \State $\QQ_{\LL} \leftarrow \LL$ with $N-1$ substituted for $N$ in all $\UB$ expressions (including for nested loops);\label{line:qli}
      \If {$\LL$ has subloops}
        \State $t \leftarrow$ nesting depth of inner-most nested loop in $\LL$;
        \State $\RR_{t+1} \leftarrow $ empty program with no statements;
	\For{$k = t; k \geq 2; k$-\xspace-} \label{line:ri}
          \For{each subloop $SL_j$  in $\LL_i$ at nesting depth $k$}  \Comment{Ordered $SL_1$, $SL_2$, $\ldots$, $SL_j$}
            \State $\RR_{SL_j} \leftarrow$ {\tt LPeel}($SL_j$, $U_{SL_j}$($\ell_1, \ldots, \ell_{k-1}, N-1$), $U_{SL_j}$($\ell_1, \ldots, \ell_{k-1}, N$));
          \EndFor
          \State $\RR_k \leftarrow$ {\tt for (i=0; i$<$$U_{\LL_{k-1}}$($N-1$); i++) \{ $\RR_{k+1}$;$\RR_{SL_1}$;$\RR_{SL_2}$;\ldots;$\RR_{SL_j}$} \}; \label{line:rk}
        \EndFor
        \State $\RR_{\LL} \leftarrow$ {\tt $\RR_2$ ; LPeel($\LL$, $U_{\LL}$($N-1$), $U_{\LL}$($N$))}; \label{line:rlli}
      \Else
        \State $\RR_{\LL} \leftarrow$ {\tt LPeel($\LL$, $U_{\LL}$($N-1$), $U_{\LL}$($N$))}; \label{line:rnonest}
      \EndIf
      \State \Return $\langle\QQ_{\LL}, \RR_{\LL}\rangle$;
    \EndProcedure
  \end{algorithmic}
\end{algorithm}



Generalizing the above intuition, Algorithm~\ref{alg:transform}
presents function \textsc{GenQandPeel} for computing $\QQ_{N-1}$ and
$\peel(\PP_N)$ for a given $\PP_N$ that has sequentially composed
loops with potentially nested loops.  Due to the
grammar of our programs, our loops are well nested.  The method works
by traversing over the structure of loops in the program.  In this
algorithm $\QQ_{\LL_i}$ and $\RR_{\LL_i}$ represent the counterparts
of $\QQ_{N-1}$ and $\peel(\PP_N)$ for loop $\LL_i$.  We create the
program $\QQ_{N-1}$ by peeling each loop in the program and then
propagating these peels across subsequent loops.  We identify the
missed iterations of each loop in the program $\PP_N$ from the upper
bound expression $\UB$.  Recall that the upper bound of each loop
${\LL_k}$ at nesting depth $k$, denoted by $U_{\LL_k}$ is in terms of
the loop counters $\ell$ of outer loops and the program parameter $N$.
We need to peel $U_{\LL_k}(\ell_1, \ell_2, \ldots, \ell_{k-1}, N) -
U_{\LL_k}(\ell_1, \ell_2, \ldots, \ell_{k-1}, N-1)$ number of
iterations from each loop, where $\ell_1 \leq \ell_2 \leq \ldots \leq
\ell_{k-1}$ are counters of the outer nesting loops.  As discussed
above, whenever this difference is a constant value, we are guaranteed
that the loop nesting depth reduces by one.  It may so happen that
there are multiple sequentially composed loops $SL_{j}$ at nesting
depth $k$ and not just a single loop $\LL_k$.  At
line~\ref{line:top-lvl}, we iterate over top level loops and call
function \textsc{GenQandPeelForLoop}($\LL_i$) for each sequentially
composed loop $\LL_i$ in $\PP_N$.  At line~\ref{line:qli}, we
construct $\QQ_{\LL}$ for loop $\LL$.  If the loop $\LL$ has no nested
loops, then the peel is the last iterations computed using the upper
bound in line~\ref{line:rnonest}.  For nested loops, the loop at
line~\ref{line:ri} builds the peel for all loops inside $\LL$
following the above intuition.  The peels of all sub-loops are
collected and inserted in the peel of $\LL$ at line~\ref{line:rlli}.
Since all the peeled iterations are moved after $Q_{\LL}$ of each
loop, we need to repair expressions appearing in $Q_{\LL}$.  The
repairs are applied by the loop at line~\ref{line:repairq}.  In the
repair step, we identify the right hand side expressions for all the
variables and array elements assigned in the peeled iterations.
Subsequently, the uses of the variables and arrays in $\QQ_{\LL_i}$
that are assigned in $\RR_{\LL_j}$ are replaced with the assigned
expressions whenever $j<i$.  If $\RR_{\LL_j}$ is a loop, this step is
more involved and hence currently not considered.  Finally at
line~\ref{line:returnqpeel}, the peels and $Q$s of all top level loops
are stitched and returned.

Note that lines~\ref{line:repairql1} and~\ref{line:repairql2} of
Algorithm~\ref{alg:transform} implement the substitution represented
by the arrow in the second column of Fig.~\ref{fig:high-lvl}.  This is
necessary in order to move the peel of a loop to the end of the
program.  If either of the loops $\LL_i$ or $\LL_j$ use array elements
as index to other arrays then it can be difficult to identify what
expression to use in $\QQ_{\LL_i}$ for the substitution.  However,
such scenarios are observed less often, and hence, they hardly impact
the effectiveness of the technique on programs seen in practice.  The
peel $\RR_{\LL_j}$, from which the expression to be substituted in
$\QQ_{\LL_i}$ has to be taken, itself may have a loop. In such cases,
it can be significantly more challenging to identify what expression
to use in $\QQ_{\LL_i}$.  We use several optimizations to transform
the peeled loop before trying to identify such an expression.  If the
modified values in the peel can be summarized as closed form
expressions, then we can replace the loop in the peel with its
summary.  For example consider the peeled loop, {\tt for ($\ell_1$=0;
  $\ell_1$<N; $\ell_1$++) \{ S = S + 1; \}}. This loop is summarized
as {\tt S = S + N;} before it can be moved across subsequent code.  If
the variables modified in the peel of a nested loop are not used
later, then the peel can be trivially moved.  In many cases, the loop
in the peel can also be substituted with its conservative
over-approximation.  We have implemented some of these optimizations
in our tool and are able to verify several benchmarks with
sequentially composed nested loops.  It may not always be possible to
move the peel of a nested loop across subsequent loops but we have
observed that these optimizations suffice for many programs seen in
practice.

\begin{theorem}
  Let $\QQ_{N-1}$ and $\peel(\PP_N)$ be generated by application of
  function \textsc{GenQandPeel} from Algorithm~\ref{alg:transform} on
  program $\PP_N$.  Then $\PP_N$ is semantically equivalent to
  $\QQ_{N-1}; \peel(\PP_N)$.
\end{theorem}

\begin{lemma}
  \label{lemma:nesting-depth}
  Suppose the following conditions hold;
  \begin{itemize}
    \item Program $\PP_N$ satisfies our syntactic restrictions (see
      Section~\ref{sec:prelims}).
    \item The upper bound expressions of all loops are linear
      expressions in $N$ and in the loop counters of outer nesting
      loops.
  \end{itemize}
  Then, the max nesting depth of loops in $\peel(\PP_N)$ is
  strictly less than that in $\PP_N$.
\end{lemma}
\begin{proof}
Let $U_{\LL_k}(\ell_1, \ldots ,\ell_{k-1}, N)$ be the upper bound
expression of a loop $\LL_k$ at nesting depth $k$.  Suppose $U_{\LL_k}
= c_1.\ell_1 + \cdots c_{k-1}.\ell_{k-1} + C. N + D$, where $c_1,
\ldots c_{k-1}, C$ and $D$ are constants.  Then $U_{\LL_k}(\ell_1,
\ldots, \ell_{k-1}, N)$ $-$ $U_{\LL_k}(\ell_1, \ldots \ell_{k-1}, N-1)
= C$, i.e. a constant.  Now, recalling the discussion in
Section~\ref{sec:genprog}, we see that {\tt LPeel($\LL_k$,
  $U_k$($\ell_1, \ldots, \ell_{k-1}, N-1$), $U_k$($\ell_1, \ldots,
  \ell_{k-1}, N$))} simply results in concatenating a constant number
of peels of the loop $\LL_k$.  Hence, the maximum nesting depth of
loops in {\tt LPeel($\LL_k$, $U_k$($\ell_1, \ldots, \ell_{k-1}, N-1$),
  $U_k$($\ell_1, \ldots, \ell_{k-1}, N$))} is strictly less than the
maximum nesting depth of loops in $\LL_k$.

Suppose loop $\LL$ with nested loops (having maximum nesting depth
$t$) is passed as the argument of function \textsc{GenQandPeelForLoop}
(see Algorithm~\ref{alg:transform}).  In line~\ref{line:ri} of
function \textsc{GenQandPeelForLoop}, we iterate over all loops at
nesting depth $2$ and above within $\LL$.  Let $\LL_k$ be a loop at
nesting depth $k$, where $2 \le k \le t$.  Clearly, $\LL_k$ can have
at most $t-k$ nested levels of loops within it.  Therefore, when {\tt
  LPeel} is invoked on such a loop, the maximum nesting depth of loops
in the peel generated for $\LL_k$ can be at most $t-k-1$.  From
lines~\ref{line:rk} and~\ref{line:rlli} of function
\textsc{GenQandPeelForLoop}, we also know that this {\tt LPeel} can
itself appear at nesting depth $k$ of the overall peel $\RR_{\LL}$.
Hence, the maximum nesting depth of loops in $\RR_{\LL}$ can be
$t-k-1+k$, i.e. $t-1$.  This is strictly less than the maximum nesting
depth of loops in $\LL$.\qed
\end{proof}

\begin{corollary}
  \label{corr:loop-free}
  If $\PP_N$ has no nested loops, then $\peel(\PP_N)$ is loop-free.
\end{corollary}

\subsection{Generating $\varphi'(N-1)$ and $\Delta \varphi'(N)$}
\label{subsec:genvarphi}
Given $\varphi(N)$, we check if it is of the form
$\bigwedge_{i=0}^{N-1} \rho_i$ (or $\bigvee_{i=0}^{N-1} \rho_i$),
where $\rho_i$ is a formula on the $i^{th}$ elements of one or more
arrays, and scalars used in $\PP_N$.  If so, we infer $\varphi'(N-1)$
to be $\bigwedge_{i=0}^{N-2} \rho_i$ (or $\bigvee_{i=0}^{N-2} \rho_i$)
and $\Delta \varphi'(N)$ to be $\rho_{N-1}$ (assuming variables/array
elements in $\rho_{N-1}$ are not modified by $\QQ_{N-1}$).  Note that
all uses of $N$ in $\rho_i$ are retained as is (i.e. not changed to
$N-1$) in $\varphi'(N-1)$.  In general, when deriving $\varphi'(N-1)$,
we do not replace any use of $N$ in $\varphi(N)$ by $N-1$ unless it is
the limit of an iterated conjunct/disjunct as discussed above.
Specifically, if $\varphi(N)$ doesn't contain an iterated
conjunct/disjunct as above, then we consider $\varphi'(N-1)$ to be the
same as $\varphi(N)$ and $\Delta \varphi'(N)$ to be True.  Thus, our
generation of $\varphi'(N-1)$ and $\Delta \varphi'(N)$ differs from
that of~\cite{tacas20}.  As discussed earlier, this makes it possible
to reason about a much larger class of pre-conditions than that
admissible by the technique of~\cite{tacas20}.

\subsection{Inferring Inductive Difference Invariants}
\label{subsec:diffinvs}

Once we have $\PP_{N-1}$, $\QQ_{N-1}$, $\varphi(N-1)$ and
$\varphi'(N-1)$, we infer \emph{difference invariants}.  We construct
the standard cross-product of programs $\QQ_{N-1}$ and $\PP_{N-1}$,
denoted as $\QQ_{N-1} \times \PP_{N-1}$, and infer difference
invariants at key control points.
Note that $\PP_{N-1}$ and $\QQ_{N-1}$ are guaranteed to have
synchronized iterations of corresponding loops (both are obtained by
restricting the upper bounds of all loops to use $N-1$ instead of
$N$).  However, the conditional statements within the loop body may
not be synchronized.  Thus, whenever we can infer that the
corresponding conditions are equivalent, we synchronize the branches
of the conditional statement.  Otherwise, we consider all four
possibilities of the branch conditions.
It can be seen that the net effect of the cross-product is executing
the programs $\PP_{N-1}$ and $\QQ_{N-1}$ one after the other.

We run a dataflow analysis pass over the constructed product graph to
infer difference invariants at loop head, loop exit and at each branch
condition. The only dataflow values of interest are differences
between corresponding variables in $\QQ_{N-1}$ and $\PP_{N-1}$.
Indeed, since structure and variables of $\QQ_{N-1}$ and $\PP_{N-1}$
are similar, we can create the correspondence map between the
variables.  We start the difference invariant generation by
considering relations between corresponding variables/array elements
appearing in pre-conditions of the two programs.  We apply static
analysis that can track equality expressions (including disjunctions
over equality expressions) over variables as we traverse the program.
These equality expressions are our difference invariants.

We observed in our experiments the most of the inferred equality
expressions are simple expressions of $N$ (atmost quadratic in
$N$). This not totally surprising and similar observations have also
been independently made in~\cite{gupta20,churchill19,barthe11}.  Note
that the difference invariants may not always be equalities.  We can
easily extend our analysis to learn inequalities using interval
domains in static analysis.  We can also use a library of expressions
to infer difference invariants using a guess-and-check framework.
Moreover, guessing difference invariants can be easy as in many cases
the difference expressions may be independent of the program
constructs, for example, the equality expression $\mathtt{v = v'}$
where $\mathtt{v} \in \PP_{N-1}$ and $\mathtt{v'} \in \QQ_{N-1}$ does
not depend on any other variable from the two programs.

For the example in Fig.~\ref{fig:high-lvl}, the difference invariant
at the head of the first loop of $Q_{N-1} \times P_{N-1}$ is
$D(V_{\QQ}, V_{\PP}, N-1) \equiv \mathtt{(x'-x = i \times (2 \times
  N-1)}$ $\wedge$ $\mathtt{\forall i \in [0,N-1),}$
  $\mathtt{a'[i]-a[i] = 1)}$, where $\mathtt{x,a} \in V_{\PP}$ and
  $\mathtt{x',a'} \in V_{\QQ}$.  Given this, we easily get
  $\mathtt{x'-x= (N-1) \times (2 \times N-1)}$ when the first loop
  terminates.  For the second loop, $D(V_{\QQ}, V_{\PP}, N-1) \equiv
  (\mathtt{\forall j \in [0,N-1), \; b'[j]-b[j] = (x'-x) + N^2 =
      (N-1) \times}$ $\mathtt{(2 \times N-1)+N^2})$.

Note that the difference invariants and its computation are agnostic
of the given post-condition.  Hence, our technique does not need to
re-run this analysis for proving a different post-condition for the
same program.

\subsection{Verification using Inductive Difference Invariants}

We present our method \textsc{Diffy} for verification of programs
using inductive difference invariants in Algorithm~\ref{alg:dpverify}.
It takes a Hoare triple $\{\varphi(N)\} \;\PP_N\; \{\psi(N)\}$ as
input, where $\varphi(N)$ and $\psi(N)$ are pre- and post-condition
formulas.  We check the base in line~\ref{line:diffy:base} to verify
the Hoare triple for $N=1$.
If this check fails, we report a counterexample.  Subsequently, we
compute $\QQ_{N-1}$ and $ \peel(\PP_N)$ as described in
section~\ref{sec:genprog} using the function \textsc{GenQandPeel} from
Algorithm~\ref{alg:transform}.  At line~\ref{line:diffy:formdiff}, we
compute the formulas $\varphi'(N-1)$ and $\Delta \varphi'(N)$ as
described in section~\ref{subsec:genvarphi}.  For automation, we
analyze the quantifiers appearing in $\varphi(N)$ and modify the
quantifier ranges such that the conditions in
section~\ref{subsec:genvarphi} hold.  We infer difference invariants
$D(V_{\QQ}, V_{\PP}, N-1)$ on line~\ref{line:diffy:dinvs} using the
method described in section~\ref{subsec:diffinvs}, wherein $V_{\QQ}$
and $V_{\PP}$ are sets of variables from $\QQ_{N-1}$ and $\PP_{N-1}$
respectively.  At line~\ref{line:diffy:psip}, we compute $\psi'(N-1)$
by eliminating variables $V_{\PP}$ from $\PP_{N-1}$ from $\psi(N-1)
\land D(V_{\QQ}, V_{\PP}, N-1)$.  At
line~\ref{line:diffy:induct-step}, we check the inductive step of our
analysis.  If the inductive step succeeds, then we conclude that the
assertion holds.  If that is not the case then, we try to iteratively
strengthen both the pre- and post-condition of $\peel(\PP_N)$
simultaneously by invoking \textsc{Strengthen}.


\begin{algorithm}[!t]
  \caption{\footnotesize \textsc{Diffy}( \{$\varphi(N)$\} $\PP_N$ \{$\psi(N)$\} )}
  \label{alg:dpverify}
  \scriptsize
  \begin{algorithmic}[1]
    \If{\;\{$\varphi(1)$\} $\PP_1$ \{$\psi(1)$\}\; fails}\label{line:diffy:base} \Comment{Base case for N=1}
      \State {\bf print} ``Counterexample found!''; \label{line:counterex}
      \State \Return $\false$;
    \EndIf

    \vspace{2mm}
    \State $\langle\QQ_{N-1}, \peel(\PP_N)\rangle \leftarrow$ \textsc{GenQandPeel}($\PP_N$);
    \State $\langle \varphi'(N-1), \Delta \varphi'(N)\rangle \leftarrow$ \textsc{FormulaDiff}($\varphi(N)$);\label{line:diffy:formdiff} \Comment{$\varphi(N) \Rightarrow \varphi'(N-1) \wedge \Delta \varphi'(N)$}
    \State $D(V_{\QQ}, V_{\PP}, N-1) \leftarrow$ \textsc{InferDiffInvs}$(\QQ_{N-1}, \PP_{N-1}, \varphi'(N-1), \varphi(N-1))$;\label{line:diffy:dinvs}
    \State $\psi'(N-1) \leftarrow$ \textsc{QE}$(V_{\PP}, \psi(N-1) \wedge D(V_{\QQ}, V_{\PP}, N-1))$;\label{line:diffy:psip}
    \If{\;\{$\psi'(N-1) \wedge \Delta \varphi'(N)$\} $\peel(\PP_N)$ \{$\psi(N)$\}\;} \label{line:diffy:induct-step}
      \State \Return $\true$;\label{line:loop-diffy-return-true} \Comment{Verification Successful}
    \Else
      \State \Return \textsc{Strengthen}$(\PP_N, \peel(\PP_N), \varphi(N), \psi(N), \psi'(N-1), \Delta \varphi'(N), D(V_{\QQ}, V_{\PP}, N))$; \label{line:diffy:strengthen}
    \EndIf

    \vspace{2mm}
    \Procedure{\textsc{Strengthen}}{$\PP_N$, $\peel(\PP_N)$, $\varphi(N)$, $\psi(N)$, $\psi'(N-1)$, $\Delta \varphi'(N)$, $D(V_{\QQ}, V_{\PP}, N)$}
    \State $\chi(N) \leftarrow \psi(N)$;\label{line:diffy:init-chi}
    \State $\xi(N) \leftarrow \true$;\label{line:diffy:init-psi}
    \State $\xi'(N-1) \leftarrow \true$;\label{line:diffy:init-psip}
    \Repeat
      \State $\chi'(N-1) \leftarrow \textsc{WP}(\chi(N), \peel(\PP_N))$;\label{line:diffy:wp}     \Comment{Dijkstra's $\mathsf{WP}$ for loop free code}
      \If{ $\chi'(N-1) = \emptyset$ }
        \If{ $\peel(\PP_N)$ has a loop }
          \State \Return \textsc{Diffy}(\{$\xi'(N-1) \wedge \Delta \varphi'(N) \wedge \psi'(N-1)$\} $\peel(\PP_N)$ \{$\xi(N) \wedge \psi(N)$\}); \label{line:loop-diffy-recursive}
        \Else
          \State \Return $\false$;\label{line:diffy:abandon1}     \Comment{Unable to prove}
        \EndIf
      \EndIf
      \State $\chi(N) \leftarrow$ \textsc{QE}$(V_{\QQ}, \chi'(N) \wedge D(V_{\QQ}, V_{\PP}, N))$;  \label{line:diffy:qe}
      \State $\xi(N) \leftarrow \xi(N) \wedge \chi(N)$;      \label{line:diffy:accumulate1}
      \State $\xi'(N-1) \leftarrow \xi'(N-1) \wedge \chi'(N-1)$;   \label{line:diffy:accumulate2}
      \If{\;\{$\varphi(1)$\} $\PP_1$ \{$\xi(1)$\}\; fails\;}  \label{line:diffy:base2}
        \State \Return $\false$;\label{line:diffy:abandon2}     \Comment{Unable to prove}
      \EndIf
      \If{\;\{$\xi'(N-1) \wedge \Delta \varphi'(N) \wedge \psi'(N-1)$\} $\peel(\PP_N)$ \{$\xi(N) \wedge \psi(N)$\}\; holds}  \label{line:diffy:induct-step2}
        \State \Return $\true$;                   \Comment{Verification Successful}
      \EndIf
    \Until{timeout};
    \State \Return $\false$;
    \EndProcedure
  \end{algorithmic}
\end{algorithm}



The function \textsc{Strengthen} first initializes the formula
$\chi(N)$ with $\psi(N)$ and the formulas $\xi(N)$ and $\xi'(N-1)$ to
$\true$.  To strengthen the pre-condition of $\peel(\PP_N)$, we infer a
formula $\chi'(N-1)$ using Dijkstra's weakest pre-condition
computation of $\chi(N)$ over the $\peel(\PP_N)$ in
line~\ref{line:diffy:wp}.  It may happen that we are unable to infer
such a formula.  In such a case, if the program $\peel(\PP_N)$ has
loops then we recursively invoke \textsc{Diffy} at line
\ref{line:loop-diffy-recursive} to further simplify the program.
Otherwise, we abandon the verification effort
(line~\ref{line:diffy:abandon1}).  We use quantifier elimination to
infer $\chi(N-1)$ from $\chi'(N-1)$ and $D(V_{\QQ},V_{\PP}, N-1))$ at
line~\ref{line:diffy:psip}.

The inferred pre-conditions $\chi(N)$ and $\chi'(N-1)$ are accumulated
in $\xi(N)$ and $\xi'(N-1)$, which strengthen the post-conditions of
$\PP_{N}$ and $\QQ_{N-1}$ respectively in
lines~\ref{line:diffy:accumulate1} - \ref{line:diffy:accumulate2}.  We
again check the base case for the inferred formulas in $\xi(N)$ at
line~\ref{line:diffy:base2}.  If the check fails we abandon the
verification attempt at line~\ref{line:diffy:abandon2}.  If the base
case succeeds, we then proceed to the inductive step.  When the
inductive step succeeds, we conclude that the assertion is verified.
Otherwise, we continue in the loop and try to infer more
pre-conditions untill we run out of time.

The pre-condition in Fig.~\ref{fig:high-lvl} is $\phi(N) \equiv \true$
and the post-condition is $\psi(N) \equiv \mathtt{\forall j \in [0,N),
    \; b[j] = j+N^3)}$.  At line~\ref{line:diffy:formdiff},
  $\phi'(N-1)$ and $\Delta \phi'(N-1)$ are computed to be $\true$.
  $D(V_{\QQ}, V_{\PP}, N-1)$ is the formula computed in
  Section~\ref{subsec:diffinvs}.  At line~\ref{line:diffy:psip},
  $\psi'(N-1) \equiv (\mathtt{\forall j \in [0,N-1), \; b'[j] = j +
      (N-1)^3 + (N-1) \times (2 \times N-1) + N^2 =}$ $\mathtt{j +
      N^3})$.  The algortihm then invokes \textsc{Strengthen} at
    line~\ref{line:diffy:strengthen} which infers the formulas
    $\chi'(N-1) \equiv \mathtt{(x'=(N-1)^3)}$ at
    line~\ref{line:diffy:wp} and $\chi(N) \equiv \mathtt{(x=N^3)}$ at
    line~\ref{line:diffy:qe}. These are accumulated in $\xi'(N-1)$ and
    $\xi(N)$, simultaneosuly strengthening the pre- and
    post-condition. Verification succeeds after this strengthening
    iteration.

The following theorem guarantees the soundness of our technique.

\begin{theorem}
  \label{thm:sound}
  Suppose there exist formulas $\xi'(N)$ and $\xi(N)$ and an integer
  $M > 0$ such that the following hold
    \begin{itemize}
    \item $\{\varphi(N)\}\; \PP_N\;\{\psi(N)\wedge \xi(N)\}$ holds for
      $1 \le N \le M$, for some $M > 0$.
    \item $\xi(N) \wedge D(V_\QQ, V_\PP, N) \Rightarrow \xi'(N)$ for
      all $N > 0$.
    \item $\{\xi'(N-1) \wedge \Delta \varphi'(N) ~\wedge~ \psi'(N-1)
      \} \; \peel(\PP_N) \; \{\xi(N) \wedge \psi(N)\}$ holds for all
      $N \ge M$, where $\psi'(N-1) \equiv \exists V_\PP \big(\psi(N-1)
      \wedge D(V_\QQ, V_\PP, N-1)\big)$.
    \end{itemize}
  Then $\{\varphi(N)\} \; \PP_N \; \{\psi(N)\}$ holds for all $N > 0$.
\end{theorem}



\section{Experimental Evaluation}
\label{sec:experiments}
We have instantiated our technique in a prototype tool called
{\ourtool}.  It is written in {\tt C++} and is built using the
LLVM(v$6.0.0$)~\cite{clang} compiler.  We use the SMT solver
{\zthree}(v$4.8.7$)~\cite{z3} for proving Hoare triples of loop-free
programs.  {\ourtool} and the supporting data to replicate the
experiments are openly available at \cite{diffy-artifact}.

\begin{table*}[!t]
\centering
\footnotesize
\begin{tabular}{|c|c|c|c|c|c|c|c|c|c|c|c|}
\hline
\multicolumn{2}{|c|}{\textsc{Program}} & \multicolumn{3}{c|}{{\ourtool}} & \multicolumn{2}{c|}{{\vajra}} & \multicolumn{2}{c|}{{\veriabs}} & \multicolumn{3}{c|}{{\viap}}\\
\multicolumn{2}{|c|}{\textsc{Category}} &
      ~~S~~ & ~~U~~ & ~TO~ & ~~S~~ & ~~U~~ & ~~S~~ & ~TO~ & ~~S~~ & ~~U~~ & ~TO~
\\ \hline \hline
Safe C1   & 110 & 110 & 0 & 0 & 110 & 0  & 96  & 14 & 16 & 1  & 93  \\ \hline
Safe C2   & 24  & 21  & 0 & 3 & 0   & 24 & 5   & 19 & 4  & 0  & 20  \\ \hline
Safe C3   & 23  & 20  & 3 & 0 & 0   & 23 & 9   & 14 & 0  & 23 & 0   \\ \hline
Total     & 157 & 151 & 3 & 3 & 110 & 47 & 110 & 47 & 20 & 24 & 113 \\ \hline \hline

Unsafe C1 & 99  & 98  & 1 & 0 & 98  & 1  & 84  & 15 & 98  & 0  & 1  \\ \hline
Unsafe C2 & 24  & 24  & 0 & 0 & 17  & 7  & 19  & 5  & 22  & 0  & 2  \\ \hline
Unsafe C3 & 23  & 20  & 3 & 0 & 0   & 23 & 22  & 1  & 0   & 23 & 0  \\ \hline
Total     & 146 & 142 & 4 & 0 & 115 & 31 & 125 & 21 & 120 & 23 & 3  \\ \hline
\end{tabular}
\caption{Summary of the experimental results. S is successful result. U is inconclusive result. TO is timeout.}
\label{tab:exp-results}
\end{table*}



{\bf Setup.}  All experiments were performed on a machine with Intel
i7-6500U CPU, 16GB RAM, running at 2.5 GHz, and Ubuntu 18.04.5 LTS
operating system.  We have compared the results obtained from
{\ourtool} with {\vajra}(v1.0)~\cite{tacas20},
{\viap}(v1.1)~\cite{viap} and {\veriabs}(v1.4.1-12)~\cite{veriabs}.
We choose {\vajra} which also employs inductive reasoning for proving
array programs and verify the benchmarks in its test-suite.  We
compared with {\veriabs} as it is the winner of the arrays
sub-category in SV-COMP 2020~\cite{svcomp20} and 2021~\cite{svcomp21}.
{\veriabs} applies a sequence of techniques from its portfolio to
verify array programs.  We compared with {\viap} which was the winner
in arrays sub-category in SV-COMP 2019~\cite{svcomp19}.  {\viap} also
employs a sequence of tactics, implemented for proving a variety of
array programs.  {\ourtool} does not use multiple techniques, however
we choose to compare it with these portfolio verifiers to show that it
performs well on a class of programs and can be a part of their
portfolio.  All tools take C programs in the SV-COMP format as input.
Timeout of 60 seconds was set for each tool.  A summary of the results
is presented in Table \ref{tab:exp-results}.

{\bf Benchmarks.}  We have evaluated {\ourtool} on a set of $303$
array benchmarks, comprising of the entire test-suite
of~\cite{tacas20}, enhanced with challenging benchmarks to test the
efficacy of our approach.  These benchmarks take a symbolic parameter
$N$ which specifies the size of each array.  Assertions are
(in-)equalities over array elements, scalars and (non-)linear
polynomial terms over $N$.  We have divided both the safe and unsafe
benchmarks in three categories.  Benchmarks in C1 category have
standard array operations such as min, max, init, copy, compare as
well as benchmarks that compute polynomials.  In these benchmarks,
branch conditions are not affected by the value of $N$, operations
such as modulo and nested loops are not present.  There are $110$ safe
and $99$ unsafe programs in the C1 category in our test-suite.  In C2
category, the branch conditions are affected by change in the program
parameter $N$ and operations such as modulo are used in these
benchmarks.  These benchmarks do not have nested loops in them.  There
are $24$ safe and unsafe benchmarks in the C2 category.  Benchmarks in
category C3 are programs with atleast one nested loop in them.  There
are $23$ safe and unsafe programs in category C3 in our test-suite.
The test-suite has a total of $157$ safe and $146$ unsafe programs.

\begin{figure}[!t]
  \begin{tabular}{cc}
    \includegraphics[scale=0.29]{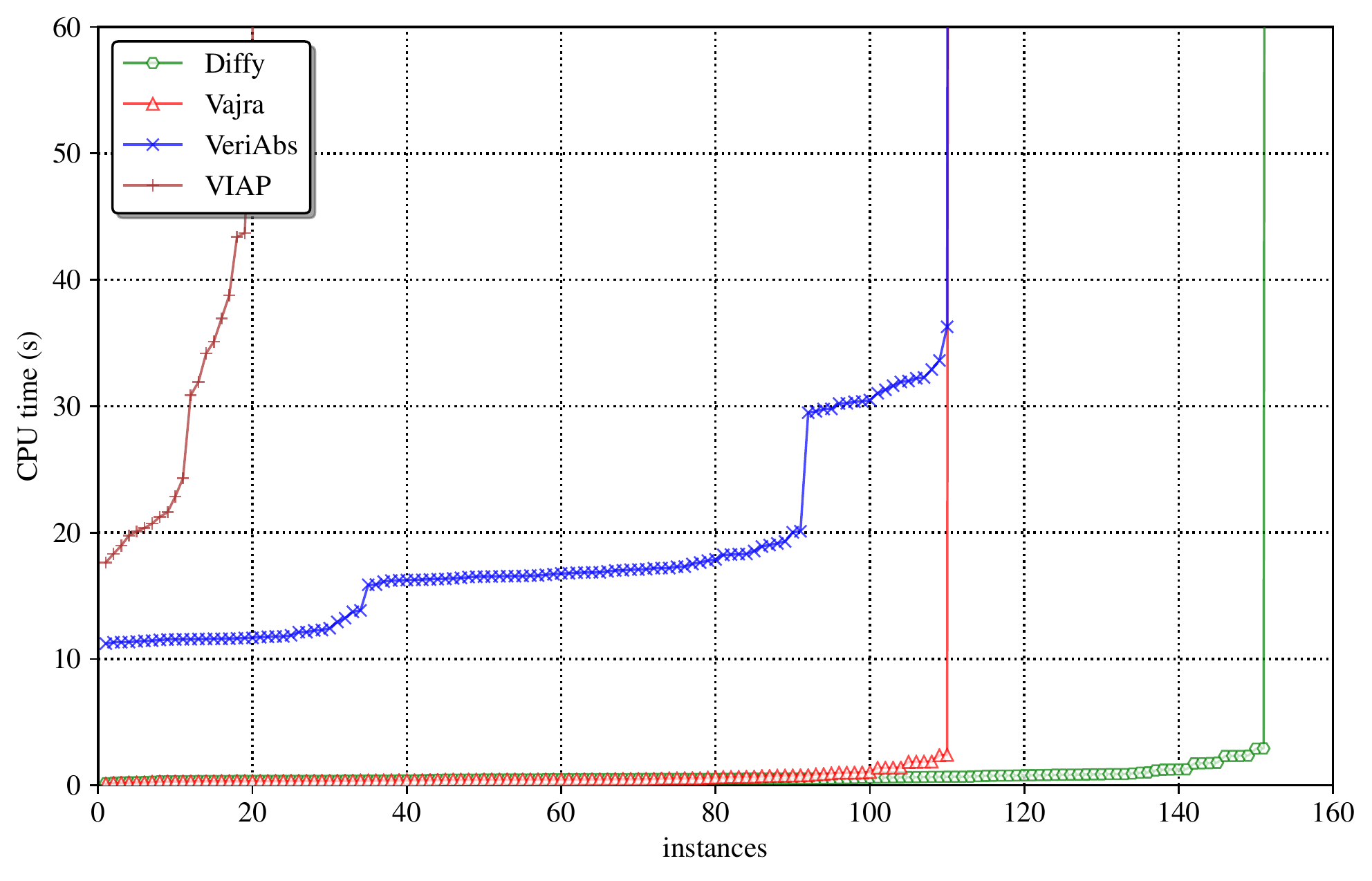}
&
    \includegraphics[scale=0.29]{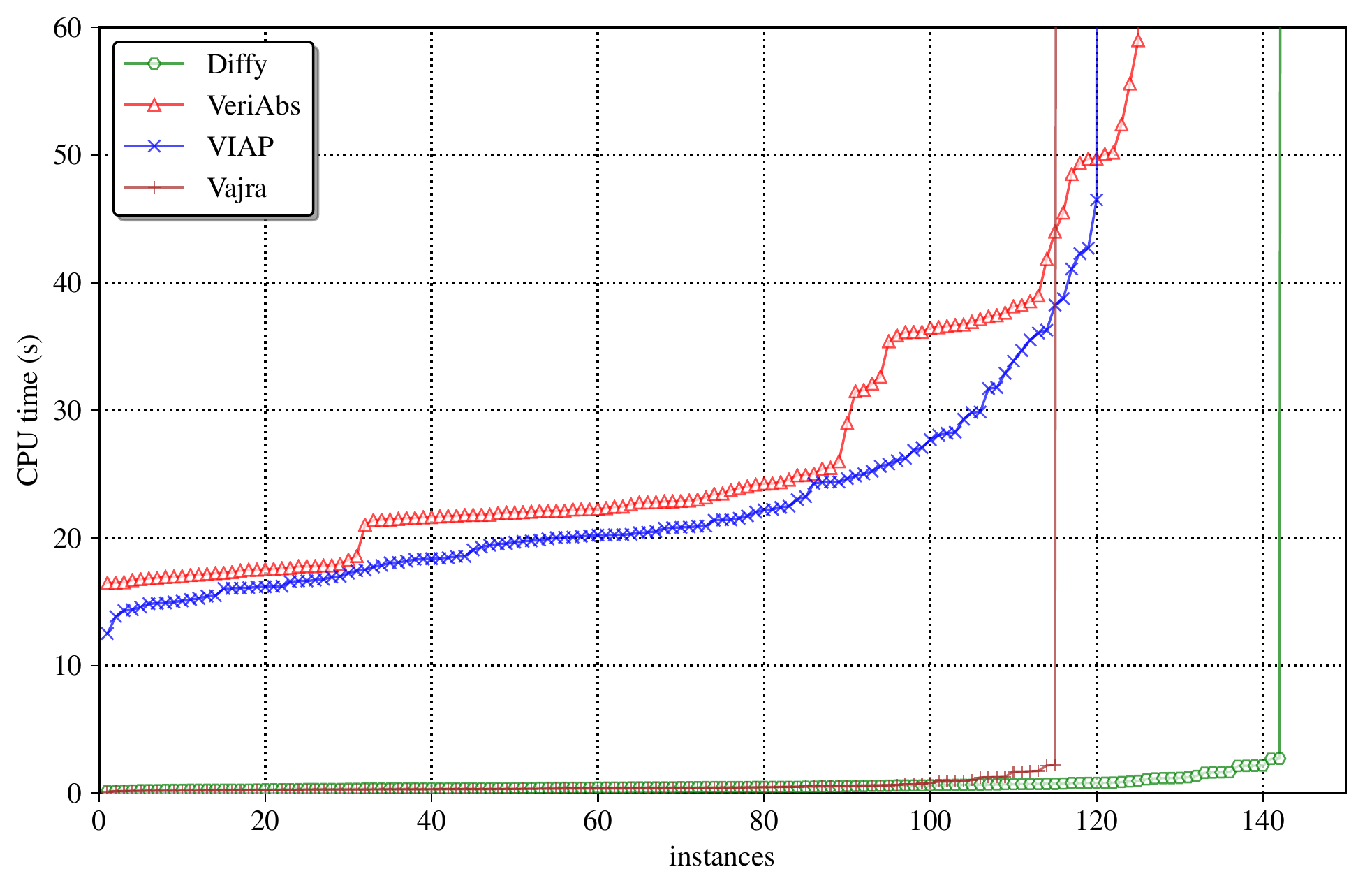}
    \\ (a) & (b)
  \end{tabular}
  \caption{Cactus Plots (a) All Safe Benchmarks (b) All Unsafe Benchmarks}
  \label{fig:cactus-all}
\end{figure}

{\bf Analysis.}  {\ourtool} verified $151$ safe benchmarks, compared
to $110$ verified by {\vajra} as well as {\veriabs} and $20$ verified
by {\viap}.  {\ourtool} was unable to verify $6$ safe benchmarks.  In
$3$ cases, the smt solver timed out while trying to prove the
induction step since the formulated query had a modulus operation and
in $3$ cases it was unable to compute the predicates needed to prove
the assertions.  {\vajra} was unable to verify $47$ programs from
categories C2 and C3.  These are programs with nested loops, branch
conditions affected by $N$, and cases where it could not compute the
difference program.  The sequence of techniques employed by
{\veriabs}, ran out of time on $47$ programs while trying to prove the
given assertion.  {\veriabs} proved $2$ benchmarks in category C2 and
$3$ benchmarks in category C3 where {\ourtool} was inconclusive or
timed out.  {\veriabs} spends considerable amount of time on different
techniques in its portfolio before it resorts to {\vajra} and hence it
could not verify $14$ programs that {\vajra} was able to prove
efficiently.  {\viap} was inconclusive on $24$ programs which had
nested loops or constructs that could not be handled by the tool.  It
ran out of time on $113$ benchmarks as the initial tactics in its
sequence took up the allotted time but could not verify the
benchmarks.  {\ourtool} was able to verify all programs that {\viap}
and {\vajra} were able to verify within the specified time limit.

\begin{figure}[!t]
  \begin{tabular}{cc}
    \includegraphics[scale=0.29]{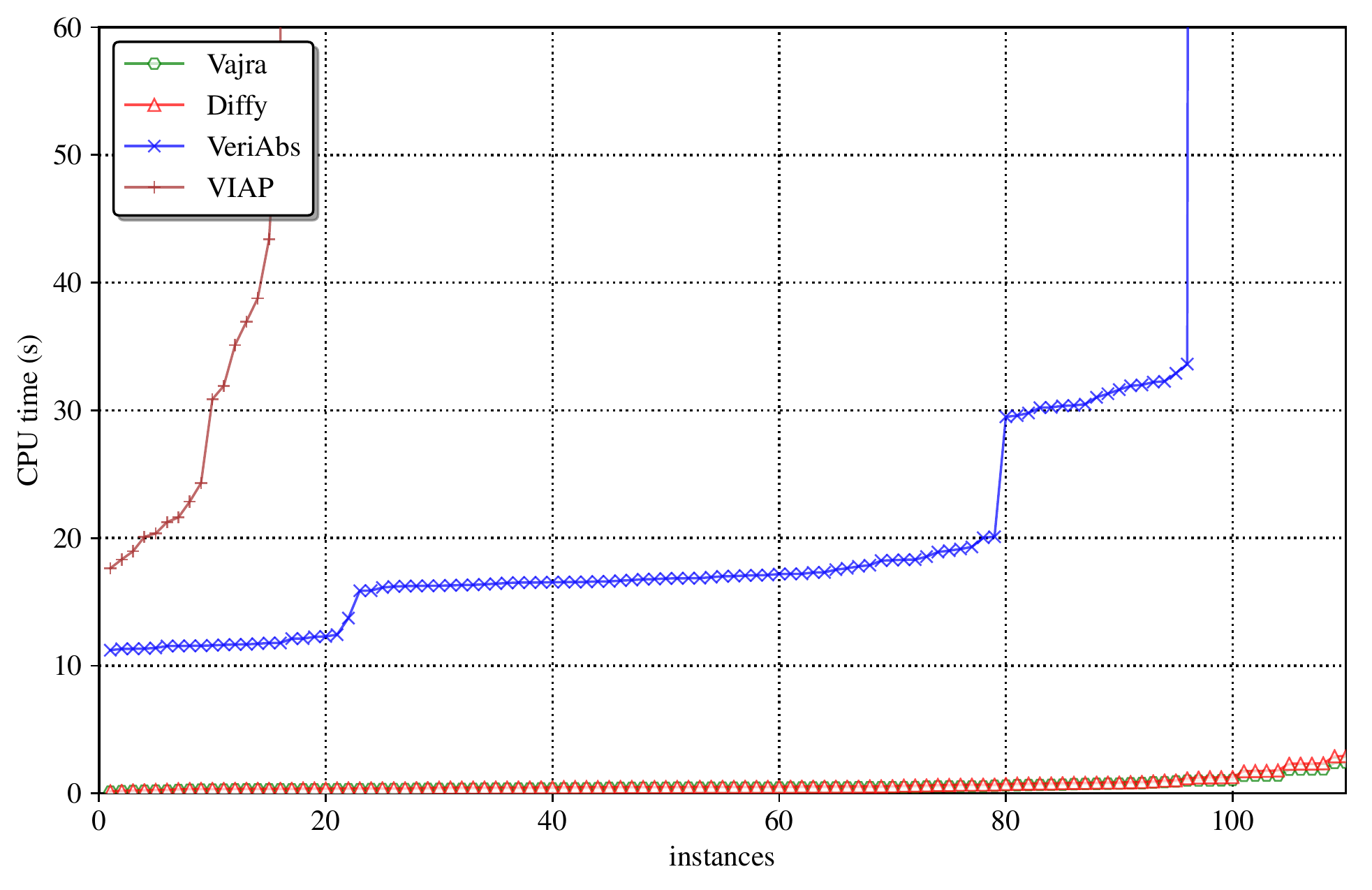}
    &
    \includegraphics[scale=0.29]{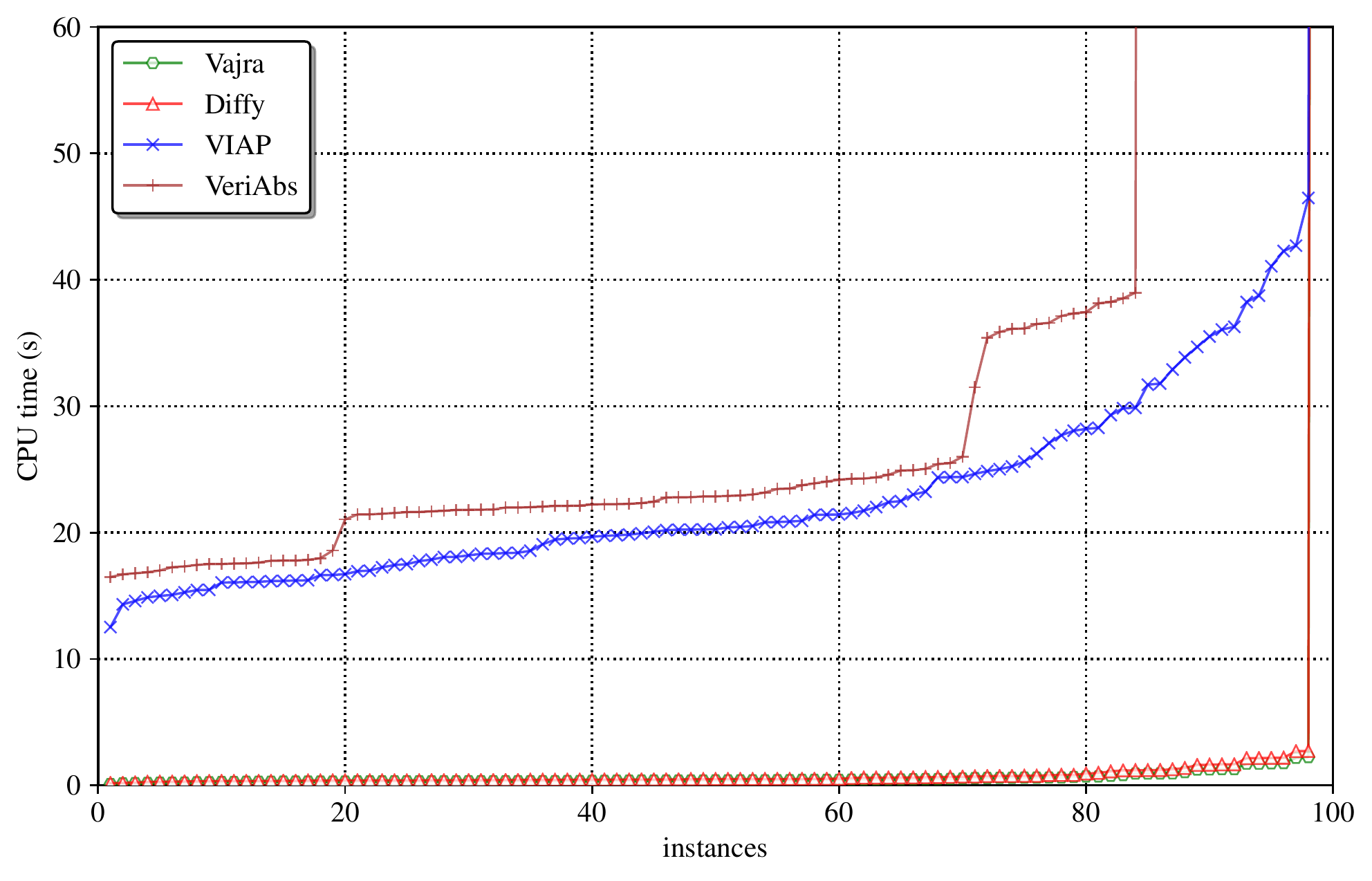}
    \\ (a) & (b)
  \end{tabular}
  \caption{Cactus Plots (a) Safe C1 Benchmarks (b) Unsafe C1 Benchmarks}
  \label{fig:cactus-C1}
\end{figure}

The cactus plot in Figure~\ref{fig:cactus-all}(a) shows the
performance of each tool on all safe benchmarks.  {\ourtool} was able
to prove most of the programs within three seconds.  The cactus plot
in Figure~\ref{fig:cactus-C1}(a) shows the performance of each tool on
safe benchmarks in C1 category.  {\vajra} and {\ourtool} perform
equally well in the C1 category.  This is due to the fact that both
tools perform efficient inductive reasoning.  {\ourtool} outperforms
{\veriabs} and {\viap} in this category.  The cactus plot in
Figure~\ref{fig:cactus-C2C3}(a) shows the performance of each tool on
safe benchmarks in the combined categories C2 and C3, that are
difficult for {\vajra} as most of these programs are not within its
scope.  {\ourtool} out performs all other tools in categories C2 and
C3.  {\veriabs} was an order of magnitude slower on programs it was
able to verify, as compared to {\ourtool}.  {\veriabs} spends
significant amount of time in trying techniques from its portfolio,
including {\vajra}, before one of them succeeds in verifying the
assertion or takes up the entire time allotted to it.  {\viap} took
$70$ seconds more on an average as compared to {\ourtool} to verify
the given benchmark.  {\viap} also spends a large portion of time in
trying different tactics implemented in the tool and solving the
recurrence relations in programs.

\begin{figure}[!b]
  \begin{tabular}{cc}
    \includegraphics[scale=0.29]{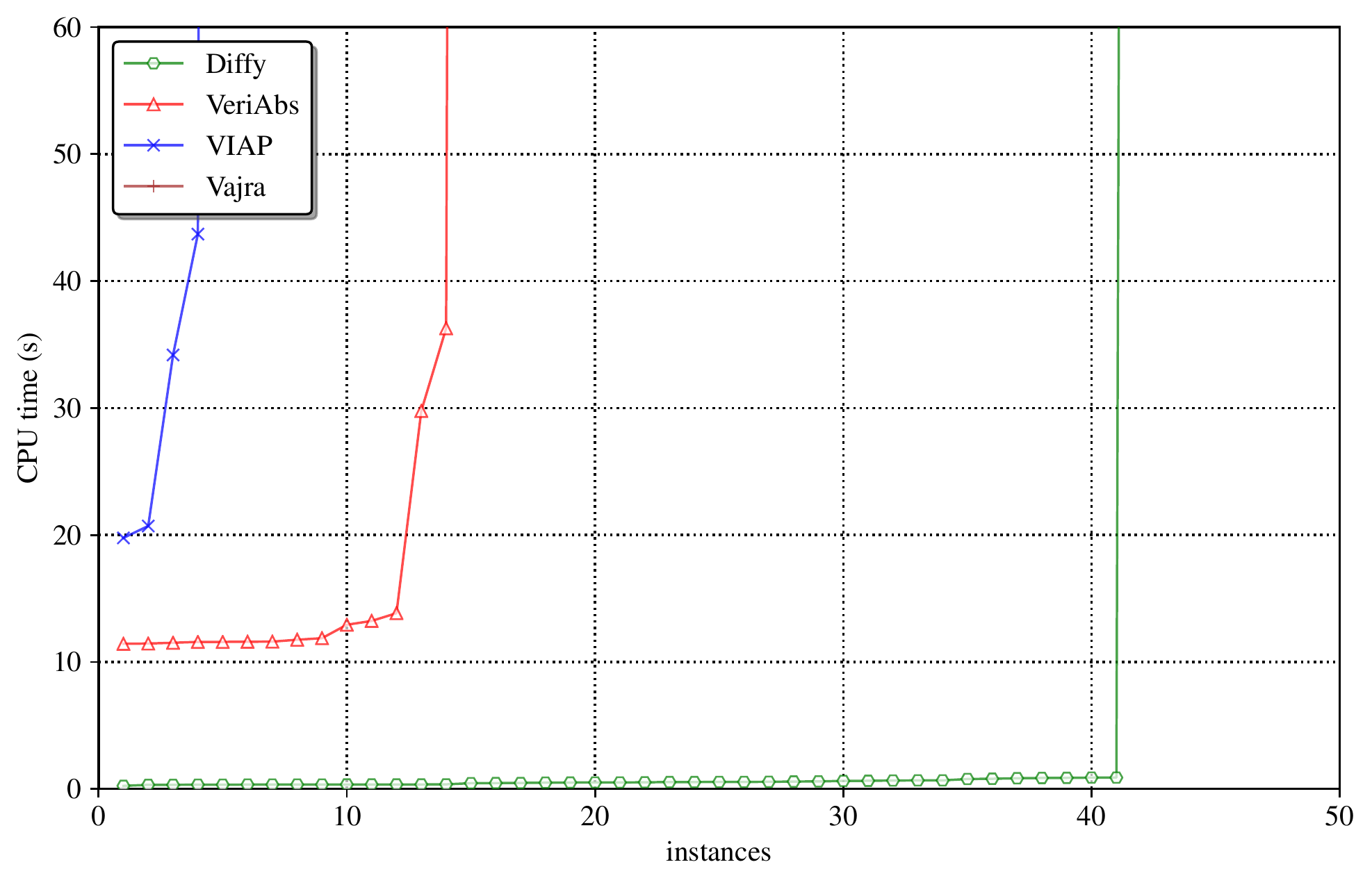}
    &
    \includegraphics[scale=0.29]{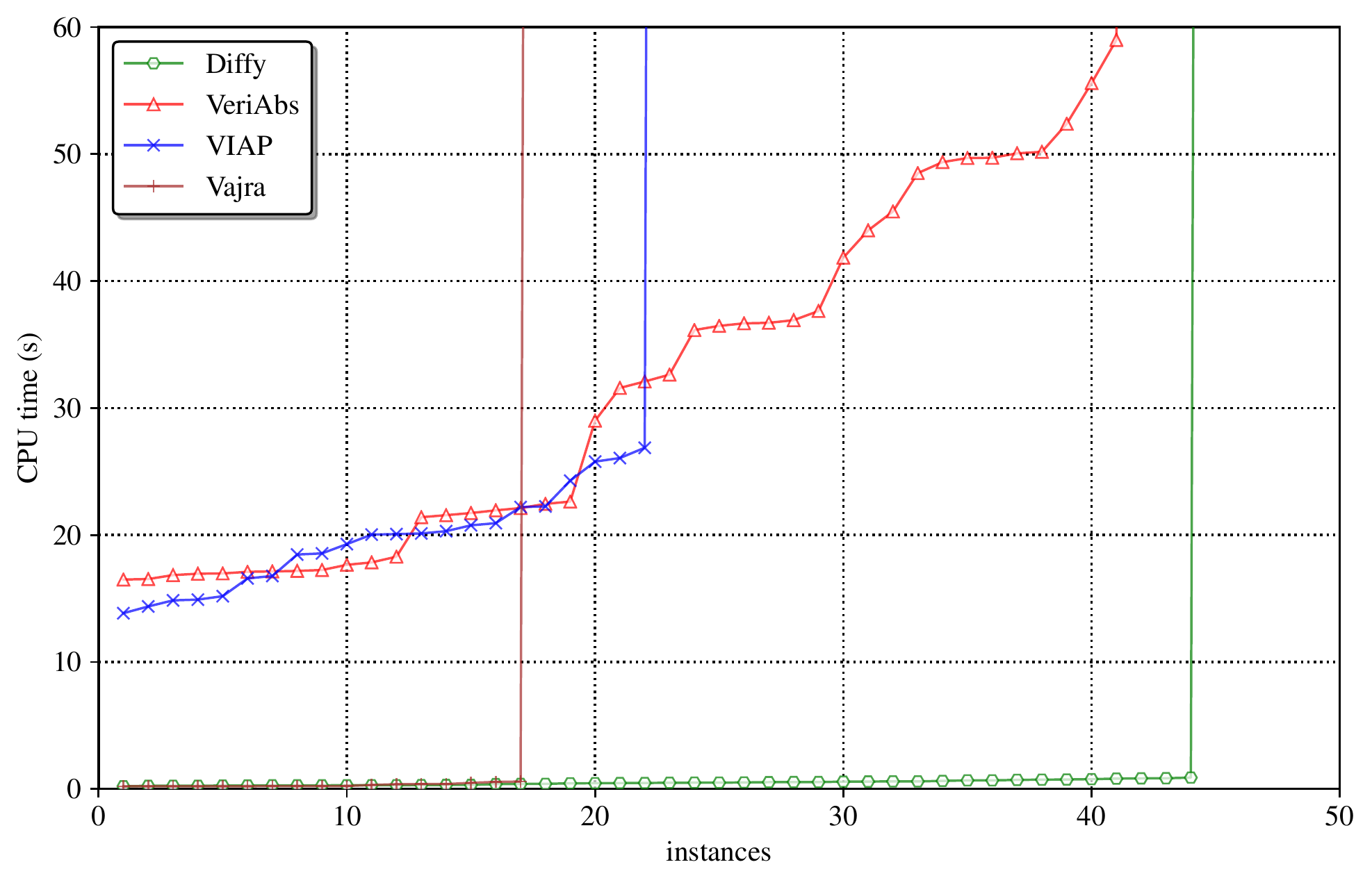}
    \\ (a) & (b)
  \end{tabular}
  \caption{Cactus Plots (a) Safe C2 \& C3 Benchmarks (b) Unsafe C2 \& C3 Benchmarks}
  \label{fig:cactus-C2C3}
\end{figure}

Our technique reports property violations when the base case of the
analysis fails for small fixed values of $N$.  While the focus of our
work is on proving assertions, we report results on unsafe versions of
the safe benchmarks from our test-suite.  {\ourtool} was able to
detect a property violation in $142$ unsafe programs and was
inconclusive on $4$ benchmarks.  {\vajra} detected violations in $115$
programs and was inconclusive on $31$ programs.  {\veriabs} reported
$125$ programs as unsafe and ran out of time on $21$ programs.
{\viap} reported property violation in $120$ programs, was
inconclusive on $23$ programs and timed out on $3$ programs.

The cactus plot in Figure~\ref{fig:cactus-all}(b) shows the
performance of each tool on all unsafe benchmarks.  {\ourtool} was
able to detect a violation faster than all other tools and on more
benchmarks from the test-suite.  Figure~\ref{fig:cactus-C1}(b) and
Figure~\ref{fig:cactus-C2C3}(b) give a finer glimpse of the
performance of these tools on the categories that we have defined.  In
the C1 category, {\ourtool} and {\vajra} have comparable performance
and {\ourtool} disproves the same number of benchmarks as {\vajra} and
{\viap}.  In C2 and C3 categories, we are able to detect property
violations in more benchmarks than other tools in less time.

To observe any changes in the performance of these, we also ran them
with an increased time out of $100$ seconds.  Performance remains
unchanged for {\ourtool}, {\vajra} and {\veriabs} on both safe and
unsafe benchmarks, and of {\viap} on unsafe benchmarks.  {\viap} was
able to additionally verify $89$ safe programs in categories C1 and C2
with the increased time limit.

\begin{figure}[!t]
  \begin{tabular}{cc}
    \includegraphics[scale=0.29]{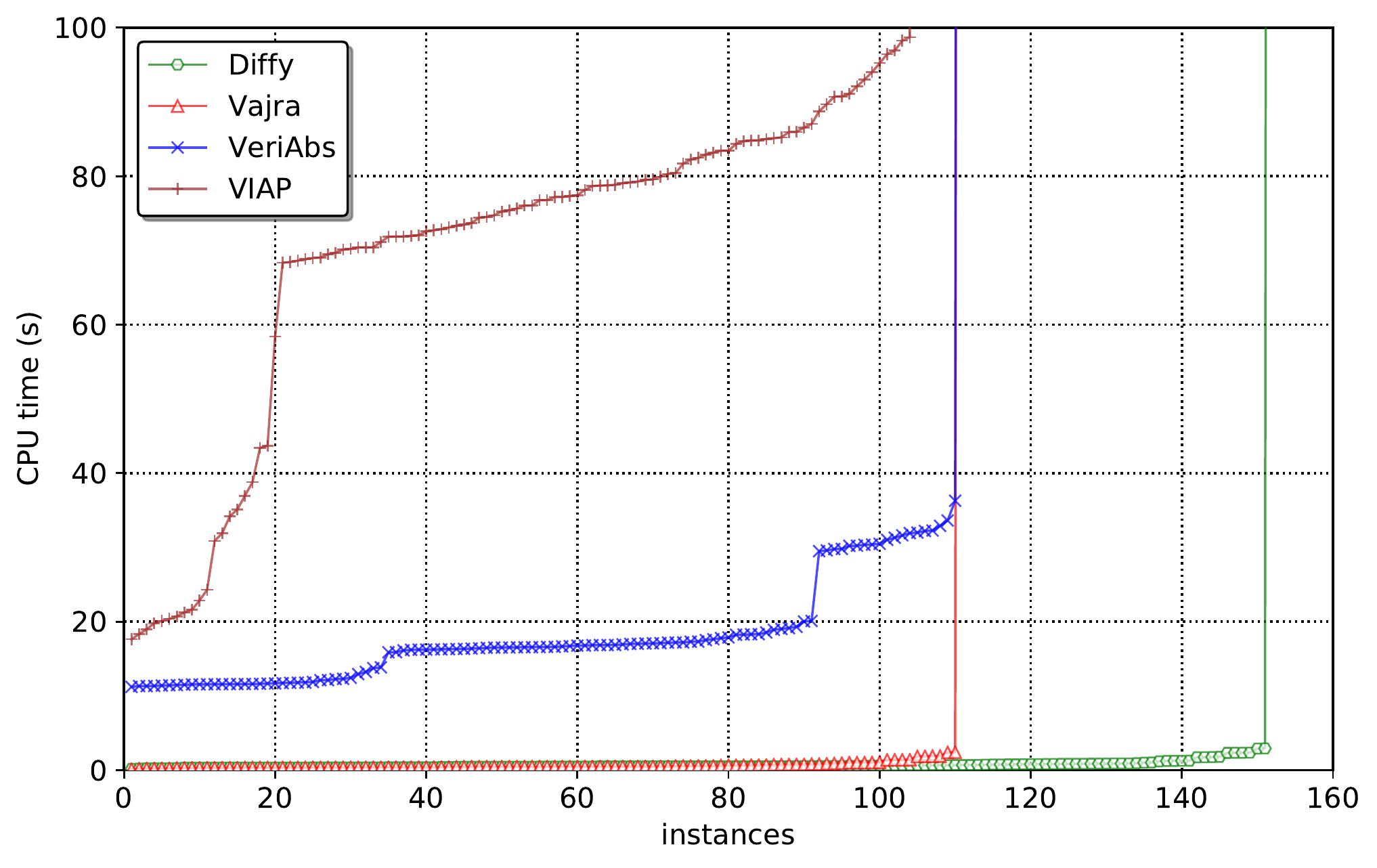}
&
    \includegraphics[scale=0.29]{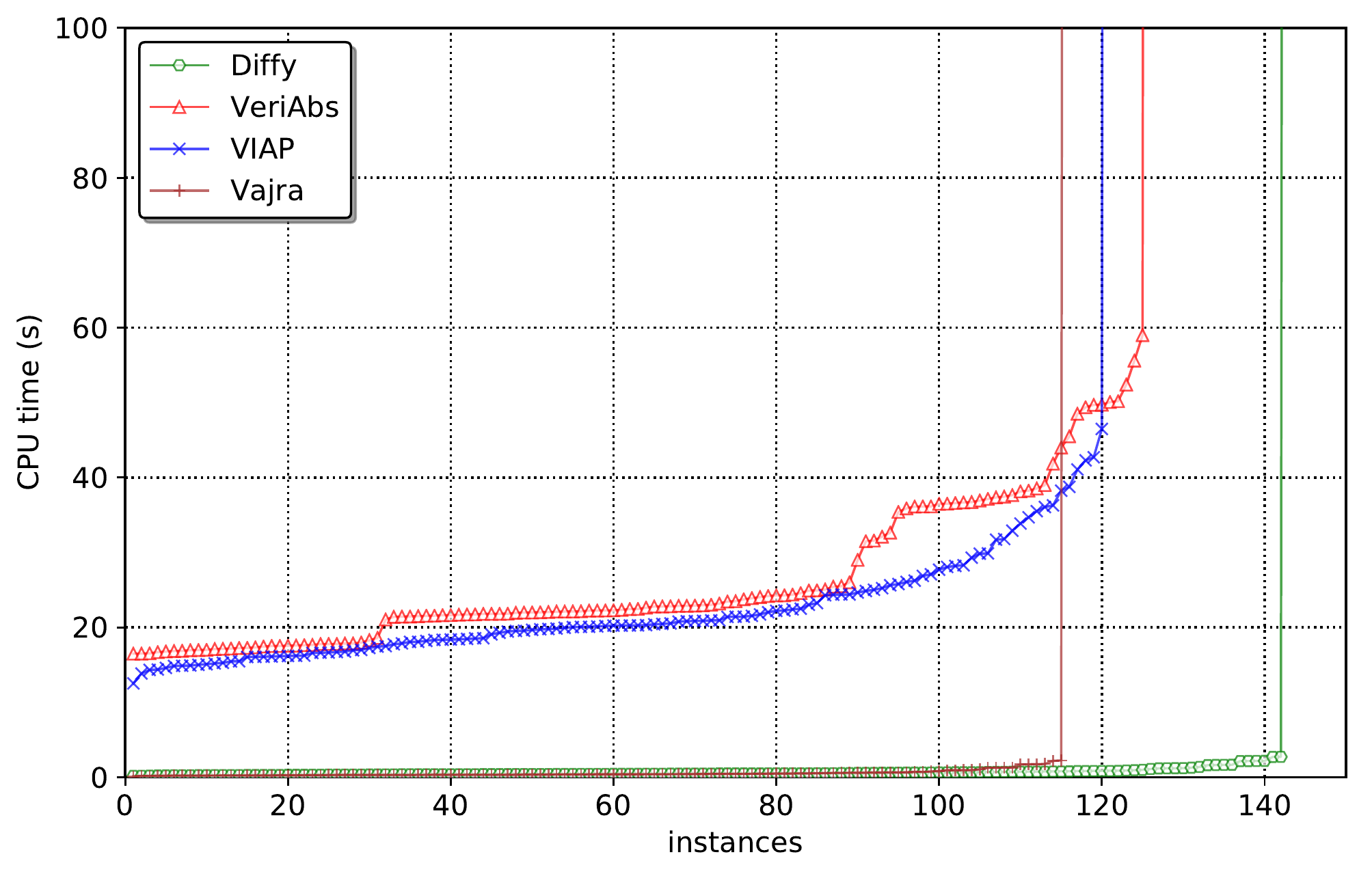}
    \\ (a) & (b)
  \end{tabular}
  \caption{Cactus Plots. TO=100s. (a) Safe Benchmarks (b) Unsafe Benchmarks}
  \label{fig:cactus-all-100}
\end{figure}



\section{Related Work}
\label{sec:related}
\textit{Techniques based on Induction.}  Our work is related to
several efforts that apply inductive reasoning to verify properties of
array programs.  Our work subsumes the full-program induction
technique in \cite{tacas20} that works by inducting on the entire
program via a program parameter $N$.  We propose a principled method
for computation and use of difference invariants, instead of computing
difference programs which is more challenging.  An approach to
construct safety proofs by automatically synthesizing squeezing
functions that shrink program traces is proposed in~\cite{squeezing}.
Such functions are not easy to synthesize, whereas difference
invariants are relatively easy to infer.  In \cite{sas17}, the
post-condition is inductively established by identifying a tiling
relation between the loop counter and array indices used in the
program.  Our technique can verify programs from \cite{sas17}, when
supplied with the \emph{tiling} relation.  \cite{brain} identifies
recurrent program fragments for induction using the loop counter.
They require restrictive data dependencies, called \emph{commutativity
  of statements}, to move peeled iterations across subsequent loops.
Unfortunately, these restrictions are not satisfied by a large class
of programs in practice, where our technique succeeds.

\textit{Difference Computation.}  Computing differences of program
expressions has been studied for incremental computation of expensive
expressions \cite{paige-differencing,liu-incrementalization},
optimizing programs with arrays \cite{liu-optimization}, and checking
data-structure invariants \cite{ditto07}.  These differences are not
always well suited for verifying properties, in contrast with the
difference invariants which enable inductive reasoning in our case.

\textit{Logic based reasoning.} In \cite{trace20}, trace logic that
implicitly captures inductive loop invariants is described.  They use
theorem provers to introduce and prove lemmas at arbitrary time points
in the program, whereas we infer and prove lemmas at key control
points during the inductive step using SMT solvers.  {\viap}
\cite{viap} translates the program to an quantified first-order logic
formula using the scheme proposed in \cite{viaptheory}.  It uses a
portfolio of tactics to simplify and prove the generated formulas.
Dedicated solvers for recurrences are used whereas our technique
adapts induction for handling recurrences.

\textit{Invariant Generation.}  Several techniques generate invariants
for array programs.  QUIC3~\cite{quic3}, FreqHorn~\cite{freqhorn} and
\cite{HornSolver} infer universally quantified invariants over arrays
for Constrained Horn Clauses (CHCs).  Template-based techniques
\cite{Gulwani,Srivastava09,Dirk07} search for inductive quantified
invariants by instantiating parameters of a fixed set of templates. We
generate relational invariants, which are often easier to infer
compared to inductive quantified invariants for each loop.

\textit{Abstraction-based Techniques.}  Counterexample-guided
abstraction refinement using prophecy variables for programs with
arrays is proposed in \cite{prophecytacas21}.
{\veriabs}~\cite{veriabs} uses a portfolio of techniques, specifically
to identify loops that can be soundly abstracted by a bounded number
of iterations.  {\vaphor}~\cite{vaphor} transforms array programs to
array-free Horn formulas to track bounded number of array cells.
{\booster}~\cite{booster} combines lazy abstraction based
interpolation\cite{lazyabsarray} and acceleration
\cite{acceleration1,acceleration2} for array programs.  Abstractions
in
\cite{ArrayCousotCL11,fluid,Gopan,Halbwachs,Jhala,Rival,Monniaux2015}
implicitly or explicitly partition the range array indices to infer
and prove facts on array segments.  In contrast, our method does not
rely on abstractions.



\section{Conclusion}
\label{sec:conc}
We presented a novel verification technique that combines generation
of difference invariants and inductive reasoning.  These invariants
relate corresponding variables and arrays from two versions of a
program and are easy to infer and prove.  These invariants facilitate
inductive reasoning by assisting in the inductive step.  We have
instantiated these techniques in our prototype {\ourtool}.
Experiments shows that {\ourtool} out-performs the tools that won the
Arrays sub-category in SV-COMP 2019, 2020 and 2021.  Investigations in
using synthesis techniques for automatic generation of difference
invariants to verify properties of array manipulating programs is a
part of future work.



\bibliographystyle{splncs04}
\bibliography{references}

\appendix
\section{Programs with $\exists$ in the Post-Condition}
In this section, we present examples with existentially quantified
formulas as post-condtions that can be verified using our technique.

\begin{figure}[h]
 \begin{tabular}{l|l}
  \begin{minipage}{0.49\textwidth}
  {\scriptsize
\begin{alltt}
// assume(true)
1. void Max(int A[], int N) \{
2.   int Max=A[0];
3.   for(i=0; i<N; i++) \{
4.     if(Max < A[i])
5.       Max = A[i];
6.   \}
7. \}
// assert(\(\exists\)i \(\in\) [0,N), Max = A[i])
\end{alltt}
  }
  \begin{center}(a)\end{center}
  \end{minipage}
  &
  \begin{minipage}{0.49\textwidth}
  {\scriptsize
\begin{alltt}
// assume(\(\forall\)i \(\in\) [0,N), A[i] > 0)
1. void Average(int A[], int N) {
2.   float B[N], sum = 0;
3.   for (int i=0; i<N; i=i+1)
4.     sum = sum + A[i];
5.   for (int j=0; j<N; j=j+1)
6.     B[j] = A[j]/sum;
7. }
// assert( sum > 0 and
//   (\(\exists\)i \(\in\) [0,N), B[i] >= 1/N) and
//   (\(\exists\)j \(\in\) [0,N), B[j] <= 1/N) )
\end{alltt}
  }
  \begin{center}(b)\end{center}
  \end{minipage}

  \\
  &
  \\
  
  \begin{minipage}{0.49\textwidth}
  {\scriptsize
\begin{alltt}
// assume(N>1)
1. void ExtGZN(int A[], int N) \{
2.   for(i=0; i<N; i++) \{
3.     A[i] =  i;
4.   \}
5. \}
// assert(\(\exists\)i \(\in\) [0,N), A[i] > 0)
// assert(\(\exists\)j \(\in\) [0,N), A[j] >= N-1)
\end{alltt}
  }
  \begin{center}(c)\end{center}
  \end{minipage}
  &
  \begin{minipage}{0.49\textwidth}
  {\scriptsize
\begin{alltt}
// assume(N>1)
1. void EvenOdd(int A[], int N) \{
2.   for(i = 0; i<N; i++) \{
3.     if(i\%2 == 0)  A[i] = 0;
4.     else  A[i] = 1;
5.   \}
6. \}
// assert(\(\exists\)i \(\in\) [0,N), A[i] = 1)
// assert(\(\exists\)j \(\in\) [0,N), A[j] = 0)
\end{alltt}
  }
  \begin{center}(d)\end{center}
  \end{minipage}

  \\
  &
  \\
  
  \begin{minipage}{0.49\textwidth}
  {\scriptsize
\begin{alltt}
// assume(\(\forall\)i \(\in\) [0,N), A[i] = 1)
1. void Sum(int A[], int N) \{
2.   int B[N], sum = 0;
3.   for (int i=0; i<N; i=i+1)
4.     sum = sum + A[i];
5.   for(int j=0; j<N; j++)
6.     B[j] = sum;
7. \}
// assert(\(\exists\)i \(\in\) [0,N), B[i] = N)
\end{alltt}
  }
  \begin{center}(e)\end{center}
  \end{minipage}
  &
  \begin{minipage}{0.49\textwidth}
  {\scriptsize
\begin{alltt}
// assume(\(\forall\)i \(\in\) [0,N), A[i] = N)
1. void Sum2(int A[], int N) \{
2.   int B[N], sum = 0;
3.   for (int i=0; i<N; i=i+1)
4.     sum = sum + A[i];
5.   for(int j=0; j<N; j++)
6.     B[j] =  sum + j;
7. \}
// assert(\(\exists\)i \(\in\) [0,N), B[i] = i + N*N)
\end{alltt}
  }
  \begin{center}(f)\end{center}
  \end{minipage}
 \end{tabular}
\caption{Examples with Existentially Quantified Post-conditions}
\label{fig:extquant}
\end{figure}

\end{document}